\documentclass[preprint,12pt]{imsart}
\pubyear{TBA}
\volume{TBA}
\issue{TBA}
\firstpage{1}
\lastpage{1}
\usepackage[dvipdfmx]{graphicx}
\usepackage{fullpage}
\usepackage{xcolor}
\usepackage{times}
\usepackage{bm}
\usepackage{natbib}
\usepackage{amsmath,amssymb,bbm}
\usepackage{amsthm}
\usepackage{enumerate}
\usepackage{appendix}
\usepackage{bbm}
\usepackage{algorithm,algorithmic}
\usepackage{framed}
\usepackage{soul}
\usepackage[normalem]{ulem}
\usepackage{url}

\startlocaldefs

\newtheorem{theorem}{Theorem}[section]

\newtheorem{lemma}[theorem]{Lemma}
\newtheorem{remark}[theorem]{Remark}

\newcommand{\R}{\mathbb{R}}

\newcommand{\Ep}{\mathbb{E}}
\newcommand{\Eppos}{\Ep_{\mathrm{pos}}}

\newcommand{\Covpos}{\mathrm{Cov}_{\mathrm{pos}}}

\newcommand{\Vpos}{\mathbb{V}_{\mathrm{pos}}}

\renewcommand{\Pr}{\mathbb{P}}
\renewcommand{\tilde}{\widetilde}
\renewcommand{\hat}{\widehat}

\newcommand{\reviseend}{\color{black}}
\newcommand{\revision}[2]{\sout{#1} {\color{blue}#2}}
\endlocaldefs

\begin{document}
\begin{frontmatter}
\title{Posterior covariance information criterion for general loss functions}
\runtitle{PCIC for general loss}

\begin{aug}
\author{\fnms{Yukito} \snm{Iba}}
\and
\author{\fnms{Keisuke} \snm{Yano}}
\address{The Institute of Statistical Mathematics, \\
10-3 Midori cho, Tachikawa City, Tokyo 190-8562, Japan. 
}
\end{aug}

\begin{abstract}
We propose a novel computationally low-cost method for estimating a general predictive measure of generalised Bayesian inference.
The proposed method utilises posterior covariance and 
provides estimators of the Gibbs and the plugin generalisation errors.
We present theoretical guarantees of the proposed method, 
clarifying the connection
to the Bayesian sensitivity analysis and the infinitesimal jackknife approximation of Bayesian leave-one-out cross validation. We illustrate several applications of our methods, including applications to differential privacy-preserving 
learning, the Bayesian hierarchical modeling, the Bayesian regression in the presence of influential observations, and the bias reduction of the widely-applicable information criterion. The applicability in high dimensions is also discussed.
\end{abstract}


\end{frontmatter}

\section{Introduction}

Bayesian statistics has achieved great successes in many applied fields because it can represent complex shapes of distributions, naturally quantify uncertainties, and accommodate the prior information commonly accepted in applied fields.
The development of the Markov chain Monte Carlo (MCMC) methods has improved the utilization of Bayesian statistics.
Nowadays, advanced MCMC algorithms are available and utilised in applied fields. Several softwares implement MCMC and enhance the accessibility of Bayesian methods. 

Once Bayesian models are fitted to the data, the goodness of fit is evaluated. Measuring predictive performance is a method for this evaluation (e.g.,~\citealp{Vehtari_Ojanen_2012}). It is also known as generalisation ability, and is an important concept in the machine learning literature. Many studies have been conducted to estimate the predictive risks of Bayesian models using specific evaluation functions.
To construct an asymptotically unbiased estimator of the posterior mean of an expected log-likelihood,
\cite{Spiegelhalter_etal_2002} employed the difference between the posterior mean of the log-likelihood and the log-likelihood evaluated at the posterior mean, and proposed the deviance information criterion (DIC).
\cite{Ando_2007} modified DIC to accommodate the model misspecification. As an alternative approach to constructing an asymptotically unbiased estimator,
\cite{Watanabe_2010} employed the sample mean of the posterior variances of sample-wise log-likelihoods, and proposed
the widely applicable information criterion (WAIC). 
As its name suggests, WAIC has now been popular in Bayesian data analysis. 
\cite{Okuno_Yano} showed that WAIC is applicable to overparametrized models.
\cite{Iba_Yano_arXiv} extended WAIC for the use in the weighted inference.
More recently, \cite{Silva_Zanella_2023} proposed yet another estimator based on the mixture representation of a leave-one-out predictive distribution.
These criteria successfully evaluate the posterior mean of the log-likelihood using samples from a single run of posterior simulation with primary data; additional simulations with leave-one-out data are unnecessary.

However, such useful tools for Bayesian models cannot accommodate arbitrary evaluation function for the prediction or arbitrary score function for the training.
Because the choice of a predictive evaluation function is application-specific and intended to gain benefits from predicting the model,
handling more general evaluation functions is desirable.
Further, generalised Bayesian approach that accommodates arbitrary score function other than the log-likelihood in the Bayesian framework has growing attentions (\citealp{Chernozhukov_Hong_2003,Jiang_Tanner_2008,Ribatet_etal_2012,Bissiri_etal_2016,Miller_Dunson_2019,Giummole_etal_2019,Nakagawa_Hashimoto_2020,Matsubara_etal_2022}). This approach makes the Bayesian approach more robust against misspecification of adequate prior distribution and model assumption.
Cross-validation (\citealp{Stone_1974,Geisser_1975}) is a ubiquitous tool for handling arbitrary evaluation function and arbitrary score function.
Although the leave-one-out cross validation (LOOCV) is intuitive and accurate, the brute force implementation requires recomputing the posterior distribution repeatedly, and is almost prohibitive. 
The importance sampling cross validation (IS-CV; \citealp{Gelfand_Dey_Chang_1992,VehtariandLampinen_2003}) estimates LOOCV without using leave-one-out posterior distributions.
The Pareto-smoothing importance sampling cross validation (PSIS-CV; \citealp{Vehtari_PSIS_2024}) extends the range where an importance sampling technique is safely used.
\cite{Vehtari_2002} proposed a generalisation of DIC to an arbitrary evaluation function.
\cite{Underhill_Smith_2016} proposed the Bayesian predictive score information criterion (BPSIC) using information matrices, and showed that BPSIC is asymptotically unbiased to generalisation error for a differentiable evaluation function.

In this study, we develop yet another novel generalisation error estimate that accommodates general evaluation and score functions in line with Bayesian computation.
The proposed method employs the \textit{posterior covariance}
to provide bias correction for empirical errors and presents an asymptotically unbiased estimator for generalisation errors. 
The proposed method demonstrates significant computational and theoretical features.
First, the proposed method can avoid the naive importance sampling technique that is sensitive to the presence of influential observations (\citealp{Peruggia_1997}); therefore, the proposed method is expected to be numerically stable with respect to such observations.
Second, the proposed method is theoretically supported by its asymptotic unbiasedness. 
Third, and most importantly, users can simply subtract the posterior covariance to estimate the arbitrary generalization error, which is rooted in an important theoretical foundation discussed in the preceding paragraph and can avoid the unstable computation of information matrices and their inverse matrices.
These advantages are illustrated 
in applications to differential privacy-preserving 
learning (\citealp{Dwork_2006}) (Subsection \ref{subsec:privacy}), the Bayesian hierarchical modeling (Subsection \ref{subsec:BHM}), the Bayesian regression in the presence of influential observations (Subsection \ref{subsec: influential observations}), the bias reduction of WAIC (Subsection \ref{subsec: bias reduction}), and the high-dimesitional models (Subsection \ref{subsec:highdimension}).

Why does the form of \textit{posterior covariance} appear in the proposal? This question leads us to find an interesting connection to the Bayesian local sensitivity and the infinitesimal jackknife approximation of the LOOCV. 
The Bayesian sensitivity formula (e.g., \citealp{Perez_etal_2006,MillarandStewart(2007),
Thomas_MacEachern_Peruggia_2018,
Giordano_etal_2018,Iba_Wkernel,Giordano_Broderick_arXiv}) implies that the posterior covariance appears by the perturbation of the posterior distribution.
This, together with
the infinitesimal jackknife approximation of LOOCV (e.g., \citealp{Beirami_etal_2017,Giordano_etal_2019,RadandMaleki_2020,Giordano_Broderick_arXiv}),
leads us to find that the proposed method corresponds to an infinitesimal jackknife approximation of the Bayesian LOOCV, which is a reason for the appearance of the posterior covariance form; see Subsection \ref{subsec: IJ}.
This aspect is reminiscent of the classical result on the asymptotic equivalence between LOOCV and Akaike information criterion (AIC; \citealp{Akaike_1973}) discovered by \cite{Stone_1974}. Also, it reduces to the asymptotic equivalence between WAIC and Bayesian LOOCV, given in Section 8 of \cite{Watanabe_book_mtbs}, when the evaluation function is the log-likelihood.

The organization of this paper is as follows.
Section \ref{sec:op} introduce our methodologies and present the infinitesimal jackknife interpretation of our methods. This interpretation gives a theoretical support (Theorem \ref{thm:PCICG}) to our methods. Section \ref{sec:example} illustrate the applications of our methods using both synthetic and real datasets.
Section \ref{sec:discussions} discuss the applicability to high-dimensional models,
the application to defining a case-influence measure, and the possibility of different definitions of the generalisation errors. Appendices include the proofs of the theorems and additional figures.

\section{Proposed method}\label{sec:op}

Suppose that we have $n$ observations $X^{n}=(X_{1},\ldots,X_{n})$ with each observation
following an unknown sampling distribution on a sample space $\mathcal{X}$ in $\R^{d}$, and
our target is to predict $X_{n+1}$ that follows the same sampling distribution
and is independent from current observations.
Let $\Theta$ be the parameter space in $ \R^{p}$.

\subsection{Predictive evaluation}

This subsection presents our methodologies.
We work with the generalised posterior distribution given by
\begin{align}
    \pi(\theta \,;\, X^{n} ) = 
    \frac{ \exp\{\sum_{i=1}^{n}s(X_{i},\theta)\}\pi(\theta) }
    {\int \exp\{\sum_{i=1}^{n}s(X_{i},\theta')\}\pi(\theta') d\theta'},
\end{align}
where $\pi(\theta)$ is a prior density on $\Theta$ and $s(x,\theta)$ is a score function that is the minus of a loss function. The score function may be different from the log-likelihood.

For the predictive evaluation of the generalised Bayesian approach,
consider an evaluation function $\nu(x,\theta)$
for a future observation $x\in\mathcal{X}$ and a parameter vector $\theta\in\Theta$.
Examples include
the mean squared error
$\nu(x,\theta)=\|x-\Ep[X\mid \theta]\|^{2}$,
the $\ell_{1}$ error
$\nu(x,\theta)=\|x-\Ep[X\mid \theta]\|$,
the log likelihood
$\nu(x,\theta)=\log p(x\mid \theta)$,
the likelihood
$\nu(x,\theta)=p(x\mid \theta)$,
and the p-value
$\nu(x,\theta)=\int_{x}^{\infty} p(x'\mid \theta)dx'$, where $\{p(x\mid\theta):\theta\in\Theta\}$ is a parametric model and $\Ep[\cdot\mid \theta]$ is the expectation with respect to $p(x\mid \theta)$.

On the basis of $\nu(x,\theta)$, 
we consider two types of predictive measures:
the Gibbs generalisation error
\[
\mathcal{G}_{\mathrm{G},n}=
\mathcal{G}_{\mathrm{G},n}(X^{n})=
\Ep_{X_{n+1}}[\,\Eppos[\nu(X_{n+1},\theta) \mid X^{n}]\,]
\]
and the plugin generalisation error
\[
\mathcal{G}_{\mathrm{P},n}=
\mathcal{G}_{\mathrm{P},n}(X^{n})=
\Ep_{X_{n+1}}[\,\nu(X_{n+1},\Eppos[\theta\mid X^{n}])
\,
],
\]
where $\Eppos[\,\,\cdot\,\,\mid\, X^{n}]$ is the generalised posterior expectation, and
$\Ep_{X_{n+1}}[\,\,\cdot\,\,]$ is the expectation with respect to $X_{n+1}$.

To estimate the Gibbs and plugin generalisation errors on the basis of current observations, we propose the Gibbs and the plugin posterior covariance information criteria $\mathrm{PCIC}_{\mathrm{G}}$ and $\mathrm{PCIC}_{\mathrm{P}}$:
\begin{align}
    \mathrm{PCIC}_{\mathrm{G}}
    &= 
    \frac{1}{n}\sum_{i=1}^{n}
    \Eppos[\nu(X_{i},\theta)\mid X^{n}]
    -\frac{1}{n}\sum_{i=1}^{n}\Covpos[\nu(X_{i},\theta),s(X_{i},\theta)\mid X^{n}] \quad \text{and}
    \\
    \mathrm{PCIC}_{\mathrm{P}}
    &= 
    \frac{1}{n}\sum_{i=1}^{n}\nu(X_{i},\Eppos[\theta\mid X^{n}])
    -\frac{1}{n}\sum_{i=1}^{n}\Covpos[\nu(X_{i},\theta),s(X_{i},\theta)\mid X^{n}],
\end{align}
where $\Covpos[\,\, \cdot \,,\, \cdot \mid X^{n}]$ is the generalised posterior covariance. 
Table 1 
summarises how we obtain estimates of the Gibbs and the plug-in generalisation errors.
The first terms correspond to the empirical errors
\begin{align*}
    \mathcal{E}_{\mathrm{G},n}=\frac{1}{n}\sum_{i=1}^{n}\Eppos[\nu(X_{i},\theta)\mid X^{n}]
    \quad\text{and}\quad
    \mathcal{E}_{\mathrm{P},n}
    &= 
    \frac{1}{n}\sum_{i=1}^{n}\nu(X_{i},\Eppos[\theta \mid X^{n}]);
\end{align*}
Our proposed methods employ the posterior covariance as bias correction of empirical errors.

\begin{table}\label{tab:riskestimate}
 \caption{Algorithm for the generalisation error estimation.}
 \begin{framed}
 \begin{algorithmic}
 \REQUIRE Observations $X_{1},\ldots,X_{n}$ and posterior samples $\theta_{1},\ldots,\theta_{M}$ from $\pi(\theta;X^{n})$ with size $M$.
 \ENSURE  An estimate of $\mathcal{G}_{\mathrm{G},n}$ and that of $\mathcal{G}_{\mathrm{P},n}$.\\ 
  \STATE Calculate
  \[T_{G}=\frac{1}{n}\sum_{i=1}^{n}\frac{1}{M}\sum_{k=1}^{M}\nu(X_{i},\theta_{k})
  \,\,\text{and}\,\,
  T_{P}=\frac{1}{n}\sum_{i=1}^{n}\nu\left(X_{i},
  \frac{1}{M}\sum_{k=1}^{M}\theta_{k}\right).\]\\
  \STATE Calculate also
  \[
  V=\frac{1}{n}\sum_{i=1}^{n}\frac{1}{M}\sum_{k=1}^{M}\{\nu(X_{i},\theta_{k})s(X_{i},\theta_{k})\}
  -
  \frac{1}{n}\sum_{i=1}^{n}
  \frac{1}{M}\sum_{k=1}^{M}\nu(X_{i},\theta_{k})
  \frac{1}{M}\sum_{k=1}^{M}s(X_{i},\theta_{k}).
  \]
 \RETURN $(T_{G}-V)$ as an estimate of $\mathcal{G}_{\mathrm{G},n}$
 and 
 $(T_{P}-V)$ as an estimate of $\mathcal{G}_{\mathrm{P},n}$.
 \end{algorithmic}
 \end{framed}
\end{table}
 
\begin{remark}[Computation]
An important point of the proposed criteria is that their computation works well with the Bayesian computation. Typical generalisation error estimates such as Takeuchi Information Criterion (TIC; \citealp{Takeuchi_1976}), Regularization Information Criterion (RIC; \citealp{Shibata_1989}), and Generalised Information Criterion (GIC; \citealp{Konishi_Genshiro_1996}) employ information matrices like $\hat{J}_{s}:=(1/n)\sum_{i=1}^{n}\{-\nabla_{\theta}\nabla_{\theta}^{\top}s(X_{i},\theta_{s})\}$ and their inverses, where $\nabla_{\theta}$ is the gradient with respect to $\theta$.
In particular, 
the computation of $\hat{J}^{-1}_{s}$ is instable and demanding in the standard Bayesian computation.
Instead of this,
the proposed method utilises posterior covariance to avoid such a demanding computation.
\end{remark}

\begin{remark}[Connection to WAIC]
When working with a parametric model $\{p(x\mid\theta):\theta\in\Theta\}$,
we set the minus log-likelihood as
the loss function,
and
consider the posterior distribution with learning rate $\beta>0$.
We then have $\nu(X,\theta)=-\log p(X\mid\theta)$
and
$s(X,\theta)=\beta \log p(X\mid \theta)$.
In this case, $\mathrm{PCIC}_{\mathrm{G}}$ is reduced to $\mathrm{WAIC}_{2}$ given in Section 8.3 of \cite{Watanabe_book_mtbs}:
\begin{align*}
    \mathrm{WAIC}_{2}=
    -\frac{1}{n}\sum_{i=1}^{n}\Eppos[\log p(X_{i}\mid\theta)\mid X^{n}]
    +\frac{\beta}{n}\sum_{i=1}^{n}\Vpos[\log p(X_{i}\mid\theta)\mid X^{n}],
\end{align*}
where $\Vpos[\,\,\cdot\,\, \mid X^{n}]$ is the generalised posterior variance.
\end{remark}

\subsection{Infinitesimal jackknife interpretation}\label{subsec: IJ}

Before presenting theoretical results, we shall discuss the infinitesimal jackknife (IJ) interpretation of the proposed method. 
The IJ approach is a general methodology that approximates algorithms requiring the re-fitting of models, such as cross validation and the bootstrap methods. In the recent machine learning literature, this methodology has been rekindled as a linear approximation of LOOCV; see \cite{Beirami_etal_2017,KohandLiang_2017,Giordano_etal_2019,RadandMaleki_2020}.
We briefly explain about the IJ approximation of the leave-one-out $Z$-estimates given as
\begin{align*}
\hat{\theta}^{(-i)}=\theta\quad\text{such that}\quad  \sum_{j\ne i} \nabla_{\theta} s(X_{j},\theta)=0\quad,i=1,\ldots,n.
\end{align*}
To describe the IJ approximation, we introduce the weighted $Z$-estimate
\begin{align*}
    \hat{\theta}_{w} := \theta \quad\text{such that}\quad  \sum_{i=1}^{n}w_{i} \nabla_{\theta} s(X_{i},\theta)=0
    \qquad \text{for}\qquad w=(w_{1},\ldots,w_{n})^{\top}\in\mathbb{R}^{n},
\end{align*}
where the leave-one-out $Z$-estimate is denoted by
\[
\hat{\theta}^{(-i)} = \hat{\theta}_{\mathbbm{1}_{-i}} \quad\text{with}\quad\mathbbm{1}_{-i}=(1,\ldots,1,\underbrace{0}_{\text{the $i$-th component}},1,\ldots,1)^{\top}.
\]
Then, by using the introduced weight $w$,
the leave-one-out estimate $\hat{\theta}^{(-i)}$ is linearly approximated as
\begin{align}
\hat{\theta}^{(-i)}_{\mathrm{IJ}} 
&= \hat{\theta} + \sum_{j=1}^{n}\frac{\partial \hat{\theta}_{w}}{\partial w_{j}}\Big{|}_{w=\mathbbm{1}}(\mathbbm{1}_{-i}-\mathbbm{1})_{j}
\quad
\text{with}\quad\mathbbm{1}=(1,\ldots,1,\ldots,1)^{\top}\nonumber\\
&=\hat{\theta}-\frac{\partial \hat{\theta}_{w}}{\partial w_{i}}\Big{|}_{w=\mathbbm{1}},
\label{eq: IJ formula}
\end{align}
which is known as the infinitesimal jackknife (IJ) approximation of leave-one-out estimates.
By the implicit function theorem,
the derivative respect to the weight is evaluated as
\begin{align*}
\frac{\partial \hat{\theta}_{w}}{\partial w_{j}}\Bigg{|}_{w=\mathbbm{1}}
=\left(-\sum_{k=1}^{n}\nabla_{\theta}\nabla_{\theta}^{\top} s(X_{k},\theta)\Big{|}_{\theta=\hat{\theta}}\right)^{-1}
\left(\nabla_{\theta}s(X_{j},\theta)\Big{|}_{\theta=\hat{\theta}}\right)
\end{align*}
and then the IJ approximation becomes
\begin{align*}
\hat{\theta}^{(-i)}_{\mathrm{IJ}} 
&=\hat{\theta}-\left(-\sum_{k=1}^{n}\nabla_{\theta}\nabla_{\theta}^{\top} s(X_{k},\theta)\Big{|}_{\theta=\hat{\theta}}\right)^{-1}
\left(\nabla_{\theta}s(X_{j},\theta)\Big{|}_{\theta=\hat{\theta}}\right).
\end{align*}

What is the Bayesian version of this formula?
Here we consider the Gibbs LOOCV
\begin{align}
\mathrm{CV}_{\mathrm{G}}
:=\frac{1}{n}\sum_{i=1}^{n}\Eppos[\nu(X_{i},\theta)\mid X^{-i}]
\label{eq: Gibbs LOOCV}
\end{align}
and 
the plug-in LOOCV
\begin{align}
\mathrm{CV}_{\mathrm{P}}
:=\frac{1}{n}\sum_{i=1}^{n}\nu(X_{i},\Eppos[\theta \mid X^{-i}])
\label{eq: plugin LOOCV}
\end{align}
where $\Eppos[\cdot\mid X^{-i}]$ is the expectation with respect to the leave-one-out generalised posterior distribution:
\[
\pi(\theta \,;\, X^{-i})
=\frac{\exp\{\sum_{j\ne i}^{}s(X_{j},\theta)\}\pi(\theta)}{\int \exp\{\sum_{j\ne i}s(X_{j},\theta')\}\pi(\theta')d\theta'}.
\]
To investigate the IJ approximation of these cross validations, we define the weighted generalised posterior distribution
\begin{align}
        \pi_{w}(\theta\,;\,X^{n}) 
    :=\frac{\exp\{\sum_{i=1}^{n}w_{i}s(X_{i},\theta)\}\pi(\theta)}{\int \exp\{\sum_{i=1}^{n}w_{i}s(X_{i},\theta')\}\pi(\theta')\mathrm{d}\theta'}\quad\text{for}\quad w=(w_{1},\ldots,w_{n})^{\top}\in\mathbb{R}^{n},
    \label{eq: weighted posterior}
\end{align}
and denote by $\Eppos^{w}[\,\,\cdot\,\,]$ the expectation with respect to $\pi_{w}(\theta\,;\,X^{n})$. Observe that we have
\begin{align*}
\Eppos[\nu(X_{i},\theta)\mid X^{-i}]
&= \Eppos^{\mathbbm{1}_{-i}}[\nu(X_{i},\theta)]
\quad\text{and}\\
\nu(X_{i},\Eppos[\theta \mid X^{-i}])
&= \nu(X_{i},\Eppos^{\mathbbm{1}_{-i}}[\theta]).
\end{align*}
This, together with employing a structure analogous to equation (\ref{eq: IJ formula}), leads us to the following IJ approximations:
\begin{align*}
(\Eppos[\nu(X_{i},\theta)\mid X^{-i}])_{\mathrm{IJ}}
&:=
\Eppos[\nu(X_{i},\theta)\mid X^{n}]
-\frac{\partial}{\partial w_{i}}\Eppos^{w}[\nu(X_{i},\theta)\mid X^{n}]\Big{|}_{w=\mathbbm{1}},\\
(\nu(X_{i},\Eppos[\theta \mid X^{-i}]))_{\mathrm{IJ}}
&:=\nu(X_{i},\Eppos[\theta \mid X^{n}])
-\frac{\partial}{\partial w_{i}}\nu(X_{i},\Eppos^{w}[\theta \mid X^{n}])\Big{|}_{w=\mathbbm{1}}.
\end{align*}

To evaluate the second terms in the IJ approximations,
we employ the following variants of local sensitivity formula \citep{MillarandStewart(2007),Giordano_etal_2018}. For more details on Bayesian local sensitivity, refer to \cite{Thomas_MacEachern_Peruggia_2018}.
\begin{lemma}\label{lem:sensitivity}
Under Condition C3 in Appendix \ref{appendix: assumptions}, we have, for $k=1,2$,
\begin{align*}
    \frac{\partial^{k}}{\partial w_{i}^{k}}
    \Eppos^{w}[\nu(X_{i},\theta)]
    &=\Eppos^{w}\left[\{\nu(X_{i},\theta)-\Eppos^{w}[\nu(X_{i},\theta)]\}\{s(X_{i},\theta)-\Eppos^{w}[s(X_{i},\theta)]\}^{k}\right]
    \\
    \frac{\partial^{k}}{\partial w_{i}^{k}}
    \Eppos^{w}[\theta]
    &=\Eppos^{w}\left[\{\theta-\Eppos^{w}[\theta]\}\{s(X_{i},\theta)-\Eppos^{w}[s(X_{i},\theta)]\}^{k}\right]
\end{align*}
In particular, we have
\begin{align*}
\frac{\partial}{\partial w_{i}}
    \Eppos^{w}[\nu(X_{i},\theta)]
    &= \Covpos^{w}[\nu(X_{i},\theta),s(X_{i},\theta)], \quad\text{and}\\
    \frac{\partial}{\partial w_{i}}
    \Eppos^{w}[\theta]
    &= \Covpos^{w}[\theta,s(X_{i},\theta)].
\end{align*}
\end{lemma}
The proof is given in Appendix \ref{appendix: proof of lemma}.
We should note that the final equation in the lemma is crucial for assessing the frequentist variance; see \cite{Giordano_Broderick_arXiv}.

For the Gibbs LOOCV, 
Lemma \ref{lem:sensitivity} yields the  form of the IJ approximation given by
\begin{align}
(\Eppos[\nu(X_{i},\theta)\mid X^{-i}])_{\mathrm{IJ}}
&=
\Eppos[\nu(X_{i},\theta)\mid X^{n}]
-\Covpos[\nu(X_{i},\theta),s(X_{i},\theta)].
\label{eq: explicit Gibbs IJ}
\end{align}
For the plug-in LOOCV,
further calculi are required.
If $\nu(X_{i},\theta)$ is 2-times continuous differentiable with respect to $\theta$
and its Hessian with respect to $\theta$ is bounded, then 
the Taylor expansion around $\hat{\theta}=\Eppos[\theta\mid X^{n}]$ yields
\begin{align*}
&\Covpos[\nu(X_{i},\theta),s(X_{i},\theta)]\\
&=\Covpos[\nu(X_{i},\theta)-\nu(X_{i},\hat{\theta})\,,\,s(X_{i},\theta)]\\
&=\Covpos\left[\sum_{a=1}^{d}\frac{\nu(X_{i},\theta)}{\partial\theta^{a}}
\Big{|}_{\theta=\hat{\theta}}
\{\theta^{a}-\hat{\theta}^{a}\}
+O_{P}(\|\theta-\hat{\theta}\|^{2})
\,,\,s(X_{i},\theta)\right]\\
&=\Covpos\left[\sum_{a=1}^{d}\frac{\nu(X_{i},\theta)}{\partial\theta^{a}}
\Big{|}_{\theta=\hat{\theta}}
\{\theta^{a}-\hat{\theta}^{a}\}
\,,\,s(X_{i},\theta)\right] + \mathrm{rem}\\
&=\sum_{a=1}^{d}\frac{\partial \nu(X_{i},\theta)}{\partial \theta^{a}}\Big{|}_{\theta=\hat{\theta}}\Covpos[\theta^{a}-\hat{\theta}^{a}\,,\,s(X_{i},\theta)]+\mathrm{rem}
\end{align*}
where $\theta=(\theta^{1},\ldots,\theta^{d})$, and $\mathrm{rem}$ is negligible under additional assumptions.
Here Lemma \ref{lem:sensitivity} gives
\[
\sum_{a=1}^{d}\frac{\partial \nu(X_{i},\theta)}{\partial \theta^{a}}\Big{|}_{\theta=\hat{\theta}}\Covpos[\theta^{a}-\hat{\theta}^{a}\,,\,s(X_{i},\theta)]
=\sum_{a=1}^{d}\frac{\partial \nu(X_{i},\theta)}{\partial \theta^{a}}\Big{|}_{\theta=\hat{\theta}}
\frac{\partial}{\partial w_{i}}\Eppos^{w}[\theta^{a}]\Big{|}_{w=\mathbbm{1}}.
\]
Observe that the chain rule gives
\[
\sum_{a=1}^{d}\frac{\partial \nu(X_{i},\theta)}{\partial \theta^{a}}\Big{|}_{\theta=\hat{\theta}}
\frac{\partial}{\partial w_{i}}\Eppos^{w}[\theta^{a}]\Big{|}_{w=\mathbbm{1}}
=\frac{\partial}{\partial w_{i}}
\nu(X_{i},\Eppos^{w}[\theta])\Big{|}_{w=\mathbbm{1}},
\]
implying that 
\[
\frac{\partial}{\partial w_{i}}
\nu(X_{i},\Eppos^{w}[\theta])\Big{|}_{w=\mathbbm{1}}
=\Covpos[\nu(X_{i},\theta),s(X_{i},\theta)]
+\mathrm{rem}.
\]
This concludes that the IJ approximation for the plug-in LOOCV is
\begin{align}
(\nu(X_{i},\Eppos[\theta \mid X^{-i}]))_{\mathrm{IJ}}
=\nu(X_{i},\Eppos[\theta \mid X^{n}])
-\Covpos[\nu(X_{i},\theta),s(X_{i},\theta)]+\mathrm{rem},
\label{eq: explicit plugin IJ}
\end{align}
where the remaining term $\mathrm{rem}$ is negligible.

So, the IJ approximations (\ref{eq: explicit Gibbs IJ}) and (\ref{eq: explicit plugin IJ}) of LOOCV presents  $\mathrm{PCIC}_{\mathrm{G}}$ and $\mathrm{PCIC}_{\mathrm{P}}$.
In the next subsection, we will show that the residual between the IJ appximations and the expected generalisation error is negligible.

In the literature on information criteria, the IJ methodology has been used to show the asymptotic equivalence between information criteria and LOOCV (\citealp{Stone_1974,Watanabe_2010_b});
see also \cite{Konishi_Genshiro_1996}.
The IJ approximation of the LOOCV $Z$-estimate requires
the second order differentiation and its inverse calculation (c.f., \citealp{Beirami_etal_2017,KohandLiang_2017}).
The discussion here emphasizes that
in Bayesian framework,
these calculi are avoidable
and there are user-friendly surrogates, that is, the posterior covariances.
Recently, \cite{Giordano_Broderick_arXiv}
analyse the Bayesian infinitesimal jackknife approximation for estimating frequentist's covariance in details, showing the consistency in finite dimensional models and giving discussions on the behaviours in high-dimensional models.

\subsection{Theoretical results}\label{subsec: theory}

On the basis of the IJ approximation discussed in the previous subsection,
this section presents the theoretical support for the use of $\mathrm{PCIC}_{\mathrm{G}}$.
The same result holds for $\mathrm{PCIC}_{\mathrm{P}}$, and it is omitted.

The following is the main theorem stating the asymptotic unbiasedness of the proposed criteria.
The proof consists of two steps: (1) the analysis on the residual between the IJ appximation and the LOOCV;
(2) the analysis on the expected values of the difference between the generalisation error and the LOOCV.
Suppose that current $n$ observations $X^{n}=(X_{1},\ldots,X_{n})$ are independent and identically distributed from an unknown sampling distribution on a sample space $\mathcal{X}$ in $\R^{d}$,
and our prediction target $X_{n+1}$ follows the same sampling distribution as that of each observation and is independent of $X^{n}$. The expectation $\Ep[\cdot]$ denotes the expectation with respect to $(X_{1,}\ldots,X_{n+1})$.

\begin{theorem}\label{thm:PCICG}
Under Conditions C1-C3,
the criterion $\mathrm{PCIC}_{\mathrm{G}}$ is asymptotically unbiased to the Gibbs generalisation error:
\begin{align}
\Ep[\mathrm{PCIC}_{\mathrm{G}}]=
\Ep[\mathcal{G}_{\mathrm{G},n}]
+o\left(\Ep[\mathcal{G}_{\mathrm{G},n}]
-\min_{\theta' \in \Theta} \mathbb{E}[\nu(X_{n+1}, \theta')]
\right).
\label{eq: in theorem}
\end{align}
\end{theorem}

\begin{proof}[Proof of Theorem \ref{thm:PCICG}]
First of all, we subtract $\min_{\theta' \in \Theta} \mathbb{E}[\nu(X_{n+1}, \theta')]$ from $\nu(X_{n+1}, \theta)$ and denote the result by $\nu$. This does not affect the proof, as we are simply subtracting the same quantity from both sides of equation (\ref{eq: in theorem}).

Next, we analyse the Gibbs LOOCV (\ref{eq: Gibbs LOOCV}) by continuing the Taylor expansion from the first order (\ref{eq: explicit Gibbs IJ}) to the second order.
Then we have
\begin{align*}
\Eppos[\nu(X_{i},\theta)\mid X^{-i}]
    =
    (\Eppos[\nu(X_{i},\theta)\mid X^{-i}])_{\mathrm{IJ}}
    +
    \frac{(-1)^{2}}{2}
    \frac{\partial^{2}}{\partial w_{i}^{2}}\bigg{|}_{w_{i}=w_{i}^{*}}
    \Eppos^{w}[\nu(X_{i},\theta)],
\end{align*}
where 
$\Eppos^{w}[\cdot]$ is defined in (\ref{eq: weighted posterior}),
$(\Eppos[\nu(X_{i},\theta)\mid X^{-i}])_{\mathrm{IJ}}$ is given in (\ref{eq: explicit Gibbs IJ}),
$w_{i}^{*}$ is a point in $[0,1]$.
Since Lemma \ref{lem:sensitivity} gives the form of the second term on the right hand side above, we get
\begin{align}
    \Eppos[\nu(X_{i},\theta)\mid X^{-i}]
    =
    (\Eppos[\nu(X_{i},\theta)\mid X^{-i}])_{\mathrm{IJ}}
    +
    \frac{1}{2}\kappa^{w^{*,i}}_{3}[i],
\label{eq:IJapproximation}
\end{align}
where 
$w^{*,i}=(1,\ldots,1,w^{*}_{i},1,\ldots,1)$ and
\begin{align*}
    \kappa^{w^{*,i}}_{3}[i]
    =\Eppos^{w^{*,i}}\left[
    \left\{\nu(X_{i},\theta)-\Eppos^{w^{*,i}}[\nu(X_{i},\theta)]\right\}
    \left\{s(X_{i},\theta)-\Eppos^{w^{*,i}}[s(X_{i},\theta)]\right\}^{2}
    \right].
\end{align*}
Summing up (\ref{eq:IJapproximation}) with respect to $i$ together with the  form (\ref{eq: explicit Gibbs IJ}) of the IJ approximation yields
the form of the Gibbs LOOCV
\begin{align}
    \mathrm{CV}_{\mathrm{G}}
    =\mathrm{PCIC}_{\mathrm{G}}
    +\frac{1}{2n}\sum_{i=1}^{n}\kappa_{3}^{w^{*,i}}[i].
\end{align}
Observe that the expected Gibbs LOOCV is just the expected Gibbs generalisation error:
\[
\Ep[\mathcal{G}_{\mathrm{G},n-1}]=\Ep[\mathrm{CV}_{\mathrm{G}}].
\]
Thus we get
\begin{align*}
    \Ep[\mathcal{G}_{\mathrm{G},n}]&=
    \Ep[\mathrm{PCIC}_{\mathrm{G}}]+
    \underbrace{
    \{
    \Ep[\mathcal{G}_{\mathrm{G},n}]
    -\Ep[\mathcal{G}_{\mathrm{G},n-1}]
    \}
    }_{=:A}
    +\underbrace{\frac{1}{2n}\sum_{i=1}^{n}\Ep[\kappa_{3}^{w^{*,i}}[i]]}_{=:B}.
\end{align*}
Condition C1 makes $A$ $o(\Ep[\mathcal{G}_{\mathrm{G},n}])$ and
Condition C2 makes $B$ $o(\Ep[\mathcal{G}_{\mathrm{G},n})]$, which completes the proof.
\end{proof}

We remark that the first part of the proof is essentially the same as the cumulant expansion used in the proof of asymptotic unbiasedness of WAIC \citep{Watanabe_book_mtbs}.

\section{Applications} \label{sec:example}

This section presents applications of the proposed methods 
in the comparison with the following existing methods:
\begin{itemize}
    \item Exact leave-one-out cross validation (LOOCV);
    \item A generalisation of DIC (GDIC; \citealp{Vehtari_2002});
    \item Bayesian predictive score information criterion (BPSIC; \citealp{Underhill_Smith_2016});
    \item Importance sampling cross validation (IS-CV; \citealp{Gelfand_Dey_Chang_1992});
    \item Pareto smoothed importance sampling cross validation (PSIS-CV; \citealp{Vehtari_PSIS_2024}).
\end{itemize}
Our applications in this section are summarised as follows:
\begin{itemize}
    \item We first present an application to the differential private learning;
    \item We second apply the methods to Bayesian hierarchical modeling;
    \item We then apply them to Bayesian regression models in the presence of influential observations;
    \item We lastly apply PCIC to eliminate the bias of WAIC due to strong priors.
\end{itemize}
Further, we remark that 
in Subsection \ref{subsec:highdimension},
we apply PCIC to high-dimensional models.

\subsection{Evaluation of differential private learning methods}\label{subsec:privacy}

\begin{figure}[h]
    \centering
    \includegraphics[width=100mm]{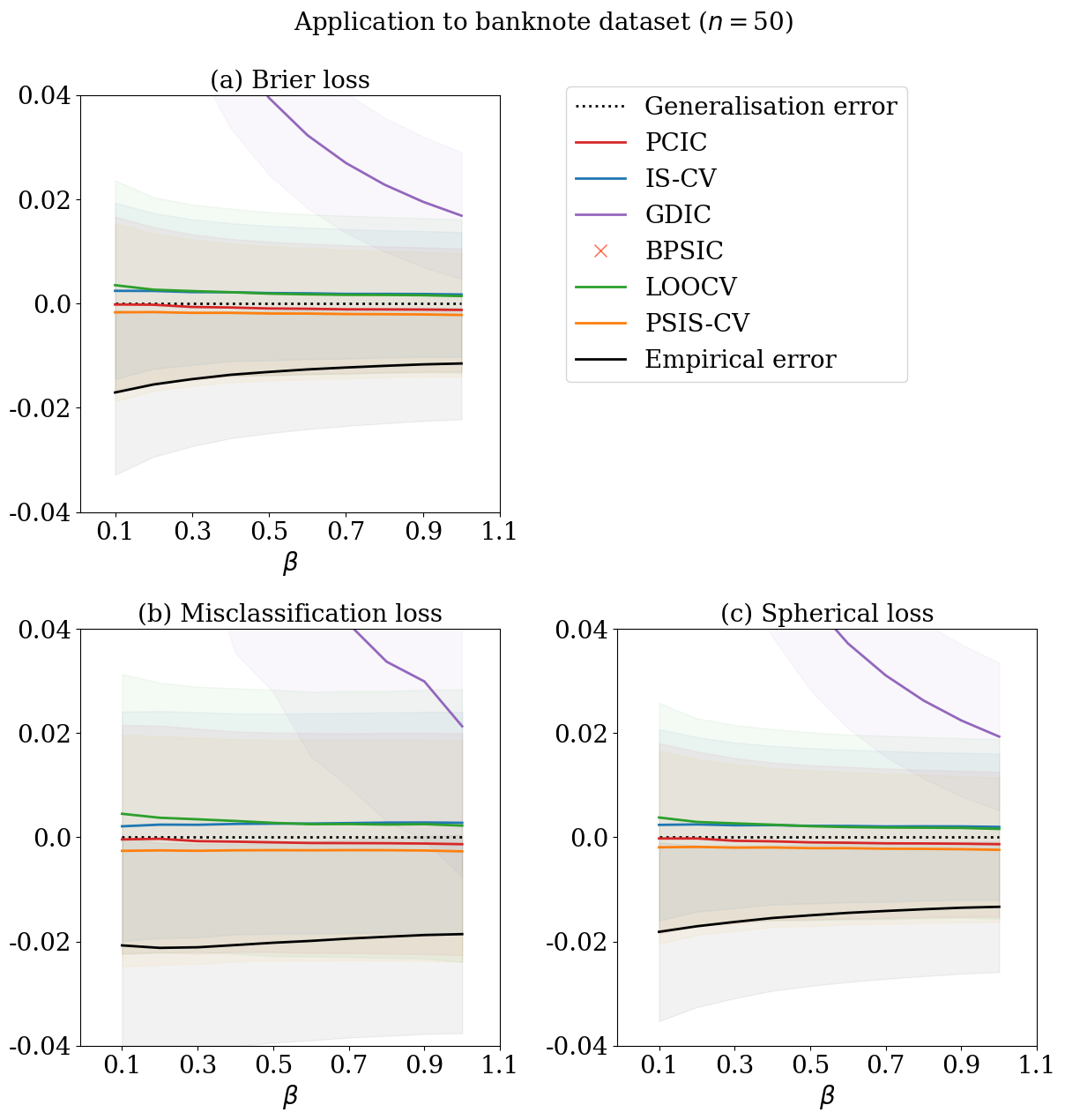}
    \caption{
    Means of generalisation error estimates relative to the generalisation errors for the banknote authentication dataset, with standard deviations represented as shaded areas. The vertical axes represent values relative to the generalisation errors, while the horizontal axes correspond to $\beta$ in the generalised posterior distribution.
(a) Generalisation error estimates for the Brier loss. (b) Estimates for the misclassification loss. (c) Estimates for the spherical loss.}
    \label{fig:banknote}
\end{figure}

\begin{figure}[h]
    \centering
    \includegraphics[width=100mm]{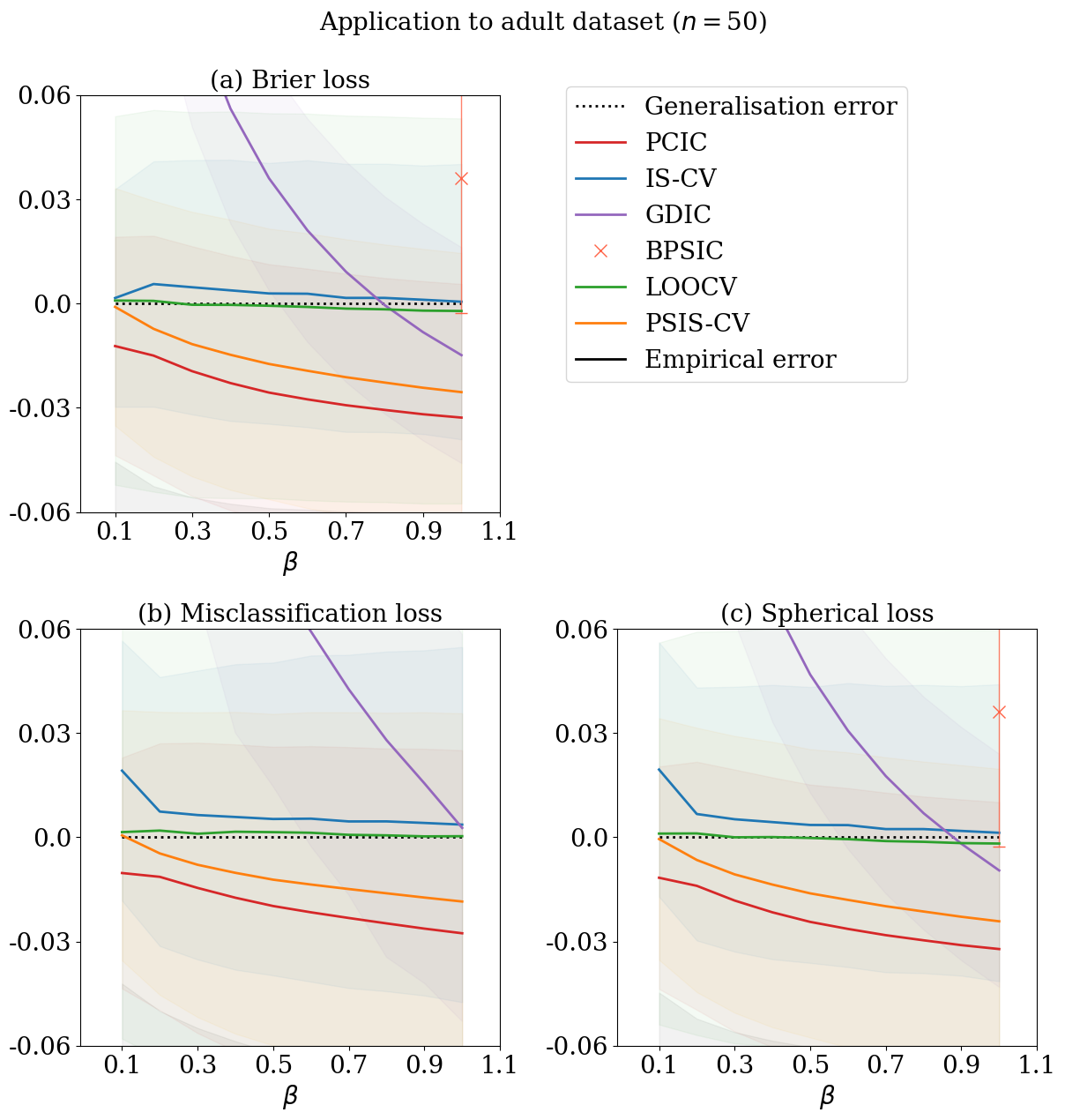}
    \caption{
    Means of generalisation error estimates relative to the generalisation errors for the adult dataset, with standard deviations represented as shaded areas. The vertical axes represent values relative to the generalisation errors, while the horizontal axes correspond to $\beta$ in the generalised posterior distribution.
(a) Generalisation error estimates for the Brier loss. (b) Estimates for the misclassification loss. (c) Estimates for the spherical loss.
    }
    \label{fig:adult}
\end{figure}

Differential privacy is a cryptographic  framework designed to ensure the privacy of individual users' data while enabling meaningful statistical analyses with a desired level of efficiency \citep{Dwork_2006}. It provides a formal guarantee that the inclusion or exclusion of a single individual's data does not significantly affect the output of the analysis, thereby protecting sensitive information. 
Differential privacy has many real-world applications (for example, see \citealp{Desfontaines_Pejo_2019}), and 
many strategies have been developed \citep{Abadi_2016,wang_PMLR_2015,Minami_2016,
Mironov_2017,Jewson_etal_2023}.

The foundational concept of ensuring differential privacy is the $(\varepsilon, \delta)$-differential private learner. For $\varepsilon, \delta \geq 0$, an $(\varepsilon, \delta)$-differentially private learner is defined as a randomized estimator $\hat{\theta}$ that satisfies the following property: for any pair of adjacent datasets $X^{n}, \tilde{X}^{n} \in \mathcal{X}^n$, that is, datasets with 
\[
\text{the number of }\{i : \text{$X_{i}$ and $\tilde{X}_{i}$ are different}\} \le 1,
\]
and for any measurable set $A$, the inequality 
\[
\Pr\left(\hat{\theta}(X^{n}) \in A\right) \leq \exp(\varepsilon) \Pr\left(\hat{\theta}(\tilde{X}^{n}) \in A\right) + \delta
\]
holds. Here, $\varepsilon$ represents the privacy budget, quantifying the level of privacy , while $\delta$ accounts for a small probability of privacy violation.

A widely employed strategy for achieving $(\varepsilon, \delta)$-differential privacy is the one-posterior-sample (OPS) estimator \citep{Wang_etal_2015,Minami_2016}. This estimator is a single sample from a generalized posterior distribution given by
\[
\pi_{\beta}(\theta; X^n) \propto \exp\left[-\beta \sum_{i=1}^n L(X_i, \theta)\right],
\]
where $L(\cdot, \cdot)$ denotes a user-specified loss or score function. The hyperparameter $\beta$ is intricately controlled by the privacy parameters $(\varepsilon, \delta)$, the choice of the score function $L$, and the dataset size $n$. This approach leverages the flexibility of posterior sampling to balance the trade-off between privacy guarantees and statistical utility.

Here, we demonstrate the application of $\mathrm{PCIC}_{\mathrm{G}}$ to understanding the predictive behaviour of OPS estimators with different values of $\beta$.
We use two sets of classification data from UCI Machine Learning Repository
(\citealp{UCI}), namely, 
the banknote authentication data set
and the adult data set.
The banknote authentication data set classifies genuine and forged banknote-like specimens based on four image features.
The adult data set 
predicts whether income exceeds $50$K/yr based on 14 features from census data.
We work with the generalised posterior distribution based on the logistic regression
\begin{align*}
    \pi_{\beta}(\theta;Y^{n},X^{n})\propto\exp\left[\beta \sum_{i=1}^{n} \left\{Y_{i}\log \sigma(X_{i}\theta)+(1-Y_{i})\log (1-\sigma(X_{i}\theta))\right\} \right]
    \exp\{-\theta^{\top}\theta/2\},
\end{align*}
where $\sigma(\cdot)$ is the sigmoid function: $\sigma(x):=1/(1+\exp(-x))$.
We consider predictive evaluation of OPS estimators using three major classification losses, the Brier loss, the misclassification loss, and the spherical loss:
\begin{align*}
    \nu_{\mathrm{Brier}}(x,p)
    &=(x-p)^{2},\\
    \nu_{\mathrm{misclass}}(x,p)
    &=\left\{
\begin{array}{ll}
-1 & \text{if}\,
x=1\,\text{and}\,p>1/2
\,\,\,\,\text{or}\,\,\,\,
x=0\,\text{and}\,p<1/2,
\\
0  &  \text{if otherwise},
\end{array}
\right.
\\
    \nu_{\mathrm{spherical}}(x,p)
    &=-\frac{\{xp+(1-x)(1-p)\}}{\sqrt{p^{2}+(1-p)^{2}}},
\end{align*}
where for the $i$ observation, $p$ is given by $\sigma(X_{i}\theta)$ using $\theta$.
The misclassification loss is non-differentiable with respect to $\theta$ and our theoretical results may not imply accurate generalisation error estimation; so, we confirm the applicability by numerical experiments.
First, we randomly split the whole data into a training dataset of a sample size 50, 20 times.
The test dataset is set to be the remaining.
We then calculate the empirical errors using the training dataset,
and the average of generalisation errors using the test dataset.
We used 3980 MCMC samples after thinning out by 5 and a burnin period of length 100. The MCMC algorithm here employs the Gibbs sampler based on Polya-Gamma augmentation \citep{polson2013}.

Figures \ref{fig:banknote} and \ref{fig:adult} display the average values of generalisation error estimates relative to  the average of the generalisation errors. The exact LOOCV is very close to the average of the generalisation error irrelevant to the loss and the dataset,
 but its computational cost is almost prohibitive.
GDIC is not close to the average of the generalisation error.
BPSIC works only for $\beta=1$ and a differentiable evaluation function.
The proposed method $\mathrm{PCIC}_{\mathrm{G}}$ successfully estimates the generalisation errors of OPS estimators for all three evaluation functions, including the non-differentiable evaluation function. 
IS-CV and PSIS-CV are comparable to the proposed method;
therefore, we cannot reach a conclusion regarding the performance difference of this setting.

\subsection{Application to Bayesian hierarchical models}
\label{subsec:BHM}

\begin{figure}[h]
    \centering
    \includegraphics[width=120mm]{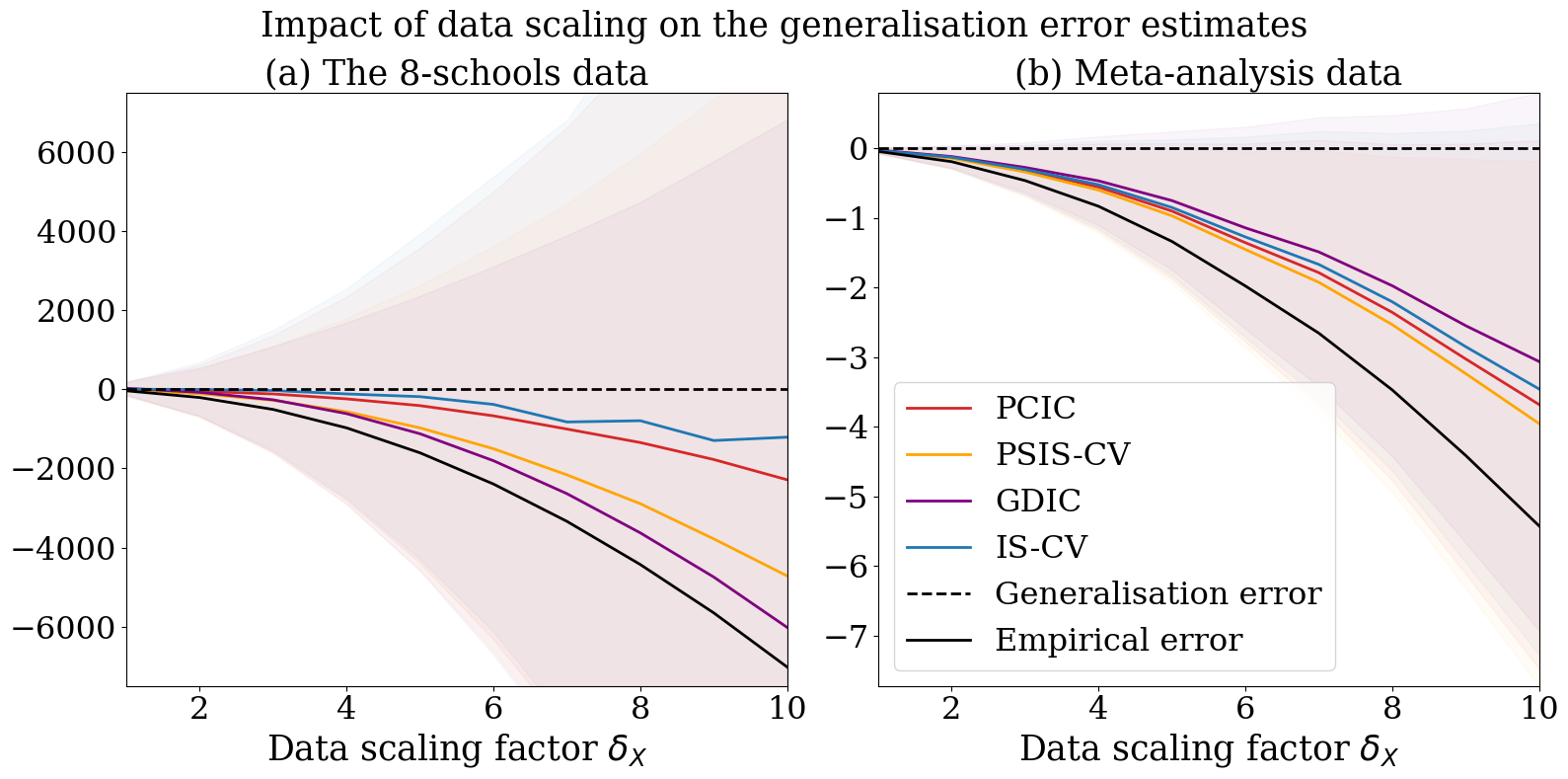}
    \caption{
    Generalisation error estimates relative to the generalisation errors with varying data scaling factors.
The averages are computed over different chains, with colored shades representing $\pm \sigma$ for these chains.
(a) Plots for the 8-schools data; (b) Plots for the meta-analysis data.
    }
    \label{fig:BHM_datascale}
\end{figure}

\begin{figure}[h]
    \centering
    \includegraphics[width=120mm]{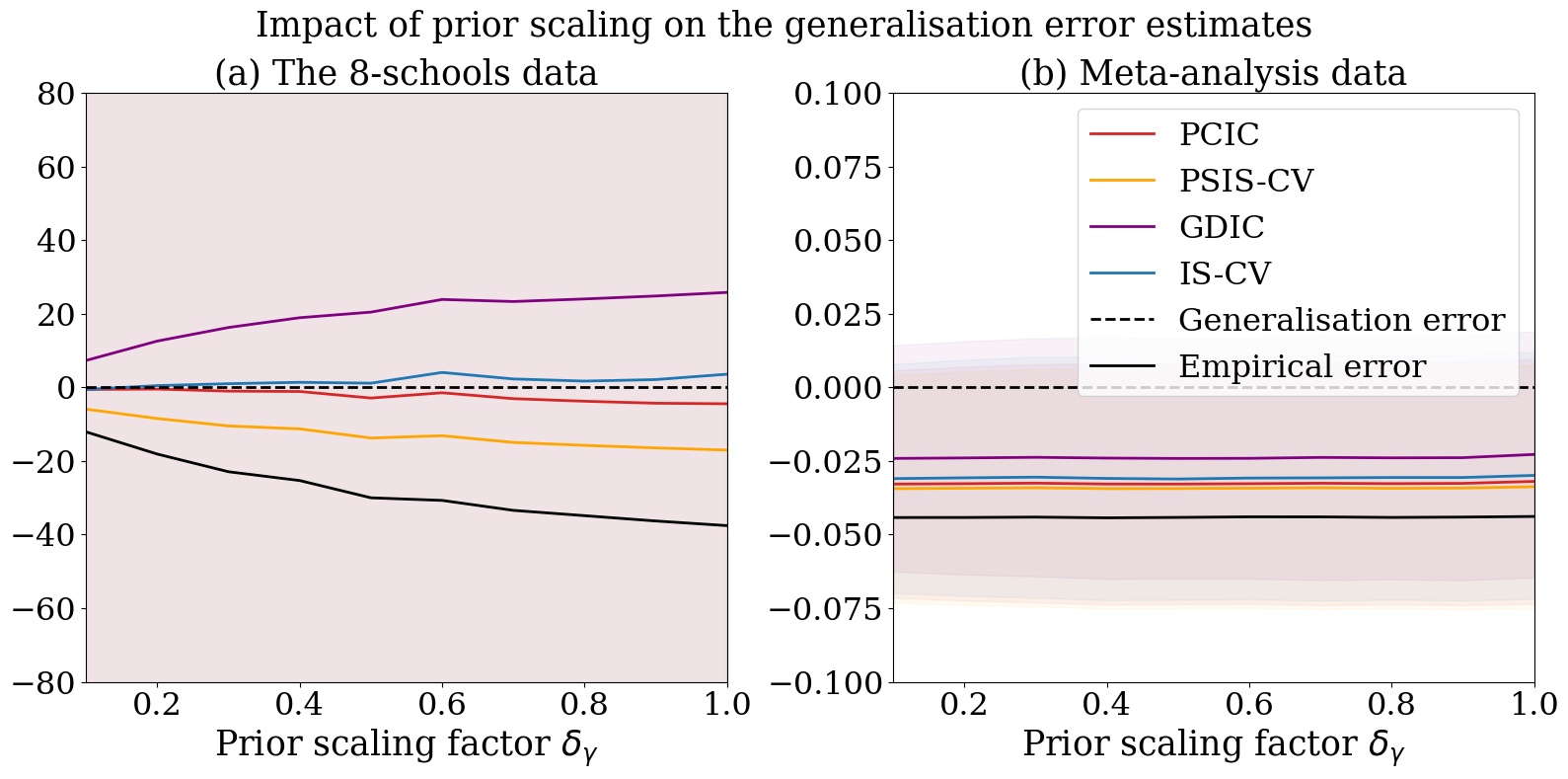}
    \caption{
    Generalisation error estimates relative to the generalisation errors with varying prior scaling factors.
The averages are computed over different chains, with colored shades representing $\pm \sigma$ for these chains.
(a) Plots for the 8-schools data; (b) Plots for the meta-analysis data.
    }
    \label{fig:BHM_priorscale}
\end{figure}

We then apply the proposed methods to 
predictive evaluation of Bayesian hierarchical modeling.
Here we shall work with a simple two-level normal model:
\begin{align*}
X_{ij} &\sim \mathcal{N}(\theta_{j},v_{j}^{2}),\quad
i=1,\ldots,n_{j},\quad
j=1,\ldots,J\\
\theta_{i} & \sim \mathcal{N}(\mu,\tau^{2}),\quad
j=1,\ldots,J,
\end{align*}
where $v_{j}^{2}$s are assumed to be given,
$n_{j}$ is the sample size of the $j$-th group ($j=1,\ldots,J$),
and $J$ is the number of groups.
Here we assume that $n_{j}$ is equal to 1 and denote $X_{ij}$ by $X_{j}$, which is always possible by the sufficiency reduction.
The prior distributions of $\mu$ and $\tau$ are specified by the scale mixture of the normal distribution with the half Cauchy distribution and
by the half normal distribution, respectively:
\begin{align*}
\mu \mid \sigma^{2}_{\mu}&\sim \mathcal{N}(0,\sigma^{2}_{\mu}),\quad
\sigma^{2}_{\mu} \sim \mathrm{Half}\mathcal{C}(\gamma),\quad\text{and}\\
\tau & \sim \mathrm{Half}\mathcal{N}(\sigma_{\tau}),
\end{align*}
where $\gamma$ and $\sigma_{\tau}$ are hyper parameters.

We measure the prediction accuracy of the information fusion in the Bayesian hierarchical structure across  groups by (unscaled) $\ell_{2}^{2}$ loss
\[
\nu(X_{j},\mu)=|X_{j}-\mu|^{2}.
\]
The working posterior in this application is rewritten as
\[
\pi(\mu,\tau ; X^{J}) \propto \exp\Bigg{(}
\sum_{j=1}^{J} s(X_{j};\mu,\tau)\Bigg{)}
\pi(\mu,\tau)
\,\,\text{with}\,\,
s_{j}(x;\mu,\tau)=\log \phi(x\,;\,\mu,\tau^{2}+v_{j}^{2}),
\]
where $\phi(x\,;\,a,b)$ is the normal density with mean $a$ and variance $b$.
So, $\mathrm{PCIC}_{\mathrm{G}}$ becomes
\[
\mathrm{PCIC}_{\mathrm{G}} = \frac{1}{J}\sum_{j=1}^{J}
\Eppos[\nu(X_{j},\mu)\mid X^{J}]-
\frac{1}{J}\sum_{j=1}^{J}\Covpos[
\nu(X_{j},\mu), s_{j}(X_{j},\mu,\tau)\mid X^{J}].
\]
To check the behaviour,
we use two datasets commonly used in the demonstration of the Bayesian hierarchical modeling.
One is the 8-schools data: the dataset consists of coaching effects with the standard deviations in $J=8$ different high schools in New Jersey; see Table 5.2 of \cite{BDA3}.
The other is the meta-analysis data from \cite{Yusuf_etal_1985}: the dataset consists of $J=22$ clinical trials of $\beta$-blockers for reducing mortality after myocardial infarction; see Table 5.4 of \cite{BDA3}. We apply empirical log-odds transformation so as to employ the normal model described above as in Section 5.8 of \cite{BDA3}.

We design two types of experiments.
One is changing the data scale by multiplying the data scaling factor $\delta_{X}\in\{1,2,\ldots,9,10\}$ to $X_{j}$s.
The other is changing the prior scale by multiplying the prior scaling factor $\delta_{\gamma}\in\{0.1i \,:\,i=1,\ldots,10\}$ to $\gamma$. 
We fix $\sigma_{\tau}=10$ and set $\gamma=10$ as initial values.
In the experiments,
we first randomly split the original datasets of $J$ groups into the training data of the remaining $\lfloor J / 2 \rfloor$ groups and the test data of $J-\lfloor J / 2 \rfloor$ groups 20 times.
We then obtain 5000 posterior samples using PyMC \citep{pymc},
calculate the empirical errors and the generalisation error estimates using the training data set,
and calculate the average of generalisation errors using the test data set.

Figure \ref{fig:BHM_datascale} displays the generalisation error estimates relative to the average generalisation errors along with the change of the data scaling factor $\delta_{X}$.
Figure \ref{fig:BHM_priorscale} displays those along with the change of the prior scaling factor $\delta_{\gamma}$.
As the data scaling factor becomes largers, all esitmates deteriorate (becomes far from the generalisation error) but the deteriorating rates of IS-CV and PCIC are slower than that of PSIS-CV. Fro any prior scaling factor, IS-CV and PCIC are closer to the generalisation error than PSIS-CV. GDIC is a relatively good metric for the meta-analysis data but not for the 8-school data.

\subsection{Application to Bayesian regression in the presence of influential observations}
\label{subsec: influential observations}

We next apply the proposed method to Bayesian regression the presence of influential observations.
The presence of influential observations impacts the variability of the case-deletion importance sampling weights, as discussed in \cite{Peruggia_1997}, which results in the instability of IS-CV.

\begin{figure}[h]
    \centering
    \includegraphics[width=120mm]{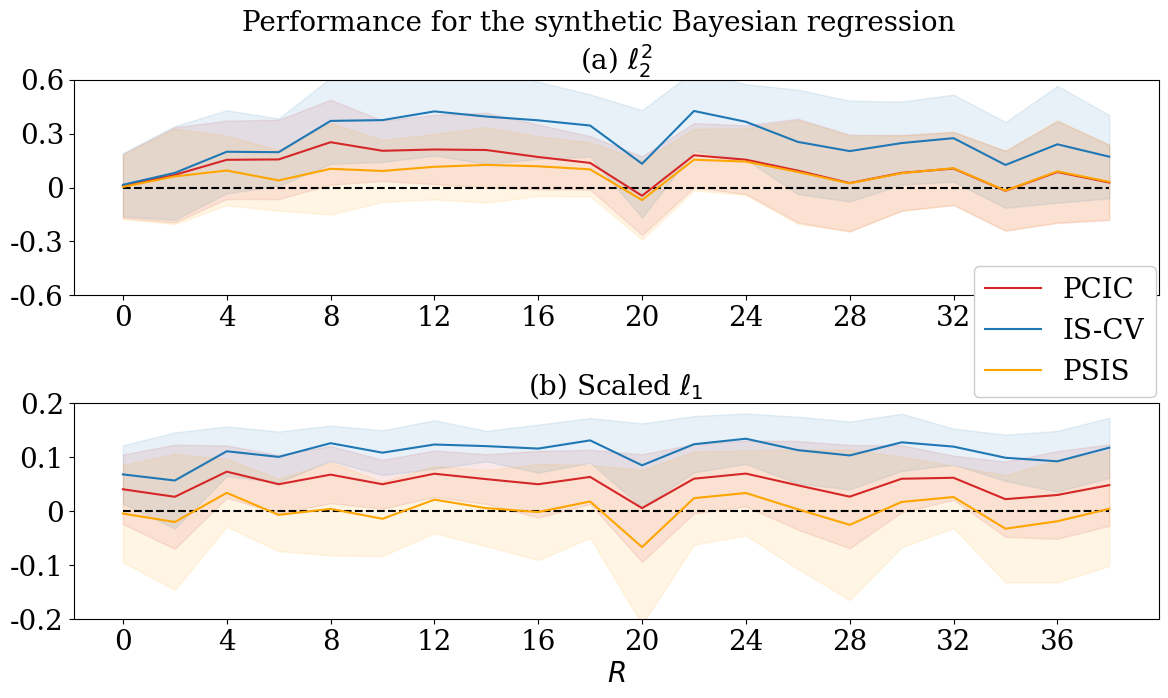}
    \caption{
    Performance of $\mathrm{PCIC}_{\mathrm{G}}$, IS-CV, and PSIS-CV for the synthetic Bayesian regression in the presence of influential observations.
    The vertical axes present values relative to the generalisation errors,
    while the horizontal axes correspond to the magnitude of the outlier.
    Averages are computed across different experiments, with colored shades representing $\pm \sigma$.
    (a) results for $\ell_{2}^{2}$ loss;
    (b) those for scaled $\ell_{1}$ loss.
    }
    \label{fig:outlier}
\end{figure}

Here, we focus on the performance comparison between our method and IS-CV in 
Bayesian regression with influential observations. 
We work with the following Bayesian regression model in \cite{Peruggia_1997}: 
Let $n$ be the sample size and
let $X_{i}\in\mathbb{R}^{d}$ ($i=1,\ldots,n$)
be the $d$ covariates. Then,
\begin{equation}
\begin{split}
Y_{i} \mid X_{i}, \beta &\sim \mathcal{N}(
X_{i}^{\top} \beta,\sigma), \,i=1,\ldots,n,\\
    \beta\mid\Sigma &\sim \mathcal{N}(0,\Sigma),\\
    \sigma^{2}      &\sim \mathrm{IG}(a_{\sigma},b_{\sigma}),\\
    \Sigma          &\sim \mathrm{IW}(\kappa\,I_{d\times d},\kappa),\label{eq:Bayesian linear regression}
    \end{split}
\end{equation}
where 
$\mathrm{IG}$ denotes the inverse gamma distribution,
$\mathrm{IW}$ denotes the inverse Wishart distribution,
and 
$I_{d\times d}$ denotes the $d\times d$ identity matrix.
Three quantities $\sigma_{a},\,\sigma_{b},\,\kappa$ are hyperparameters, and
we fix $\sigma_{a}=\sigma_{b}=\kappa=1$.
The score function $s(x,\theta)$ in this subsection is simply the log-likelihood.
For the loss functions, we use 
$\ell_{2}$, $\ell_{1}$, scaled $\ell_{2}$ and scaled $\ell_{1}$ losses:
\begin{align*}
    \nu_{\ell_{2}}(Y_{i},X_{i},\beta,\sigma)&=|Y_{i}-X_{i}^{\top}\beta|^{2}
    \quad\text{and}\quad
    \nu_{\ell_{1}}(Y_{i},X_{i},\beta,\sigma)=|Y_{i}-X_{i}^{\top}\beta|,\\
    \nu_{\text{scaled}\, \ell_{2}}(Y_{i},X_{i},\beta,\sigma)&=|Y_{i}-X_{i}^{\top}\beta|^{2}/\sigma^{2},
    \quad\text{and}\quad
    \nu_{\text{scaled}\, \ell_{1}}(Y_{i},X_{i},\beta,\sigma)=
    |Y_{i}-X_{i}^{\top}\beta|/\sigma.
\end{align*}

We begin with the following synthetic Bayesian regression model: for a given $R$,
$\tilde{X}_{1},\ldots,\tilde{X}_{n}$ are fixed to
\begin{align*}
    \tilde{X}_i &= \begin{cases}
    0.01i & \text{if $i<n$},\\
    R & \text{if $i=n$},
           \end{cases}
\end{align*}
and then $Y_{1},\ldots,Y_{n}$ are given as
\begin{align*}
    Y_i &= \beta_0 + \tilde{X}_i \beta_1 + \sigma\varepsilon_i,
\end{align*}
where $\varepsilon_i$s are i.i.d.~from the standard Gaussian distribution, and $\beta_{0}$, $\beta_{1}$, and $\sigma$ are unknown parameters. We assign (0,1,1) to the true values of $(\beta_{0},\beta_{1},\sigma)$. 
We set 50 to the sample size $n$. 
For each $R$,
we simulate the training data set 50 times and obtain 50 values of $\mathrm{PCIC}_{\mathrm{G}}$, PSIS-CV, and IS-CV.
We then calculate the average of the generalisation errors using a test data set with a sample size of 10. 
 By using the same Gibbs sampler as in \cite{Peruggia_1997},
we obtained 2000 MCMC samples after thinning out by 10 and a burnin period of length 10000.

Figure \ref{fig:outlier} displays the comparison between $\mathrm{PCIC}_{\mathrm{G}}$ and IS-CV.
Consider $\nu_{\ell_{2}}$.
For $R \le 2$, 
the performance of the two methods does not differ.
For $R \ge 3$,
the bias of IS-CV relative to the average of the generalisation errors, is larger than that of $\mathrm{PCIC}_{\mathrm{G}}$.
Consider $\nu_{\mathrm{scaled}\,\ell_{1}}$.
In this case, 
for all $R$,
the bias of IS-CV, relative to the average of the generalisation errors, is larger than those of $\mathrm{PCIC}_{\mathrm{G}}$ and PSIS-CV.
In particular, PSIS-CV performs the best for this loss.

We next employ two real datasets containing influential observations.
One is the the stack loss data:
The data consists of $n=21$ days of measurements with three condition records $x_{i}$.
The other is the Gesell data:
This consists of $n=21$ children's records of the age $x_{i}$ at which they first spoke 
and their Gesell Adaptive Score $y_{i}$. These datasets are known to contain influential observations; 
the indices for influential observations are
$i=21$ and $i=18$, respectively.
In the experiments,
we first randomly split the original datasets into the training data of $\lfloor n / 2 \rfloor$ samples and the test data of $n-\lfloor n / 2 \rfloor$ samples 20 times.
We then obtain posterior samples with different sample sizes,
calculate the empirical errors and the generalisation error estimates using the training data set,
and calculate the average of generalisation errors using the test data set.

Figures \ref{fig:stackloss} and \ref{fig:Gesell} display the results.
Overall, 
$\mathrm{PCIC}_{\mathrm{G}}$ and IS-CV works better than PSIS-CV for the unscaled losses ($\ell_{1}$ and $\ell_{2}$ losses),
whereas PSIS-CV works the best for the scaled losses.
One of the reasons is that 
the scaled losses works similar (in fact, the scaled $\ell_{2}$ loss is the half of log likelihood),
and the weight smoothing in PSIS-CV only depends only on log-likelihood.

\subsection{Application to eliminating biases due to strong priors}
\label{subsec: bias reduction}

When using strong priors, 
WAIC is shown to have a bias in the generalisation error estimation
(e.g., \cite{Vehtari_etal_2017,Ninomiya_2021}).
We here employ $\mathrm{PCIC}_{\rm G}$ to eliminate this bias.
We begin with focusing on a simple example that enables a full analytic calculus and then provide a general scheme applicable to general models.

\textbf{Analytic calculus in a location-shift model}: 
Consider a simple location-shift model, where the observations $X^{n} = (X_{1},\ldots,X_{n})$ follow 
\begin{align*}
    X_{i}=\theta^{*}+\varepsilon_{i},\,\,i=1,\ldots,n.
\end{align*}
Here,
$\theta^{*}$ is a vector in $\mathbb{R}^{d}$,
and
the error terms $\varepsilon_{1},\ldots,\varepsilon_{n}$ are i.i.d.~from a possibly non-Gaussian distribution with mean zero and covariance matrix identity matrix.
Consider the generalised posterior distribution given by
\begin{align*}
    \pi(\theta; X^{n})\propto \exp{\left\{-\beta\frac{\sum_{i=1}^{n}\|X_{i}-\theta\|^{2}}{2} - \frac{\|\theta\|^2}{2\tau} \right\}} \quad \text{with}\quad 
    \beta>0\,\text{and}\,\tau>0,
\end{align*}
where $\|\cdot\|$ is the $\ell_{2}$ norm in $\mathbb{R}^{d}$;
then, the generalised posterior distribution is $\mathcal{N}(\hat{\theta},S)$ with
\begin{align*}
    \hat{\theta}=\frac{n\beta\tau}{n\beta\tau+1}\overline{X}
    \quad\text{and}\quad
    S=\frac{1}{n\beta+1/\tau}I_{d}
\end{align*}
and the score function in this case is $s(x,\theta)=-(\beta/2)\|x-\theta\|^{2}$,
where
let $\overline{X}:=\sum_{i=1}^{n}X_{i}/n$ and let $I_{d}$ be the $d\times d$ identity matrix.

Consider the evaluation function $\nu(x,\theta)=(x-\theta)^{\top}A(x-\theta)$ with a symmetric positive definite matrix $A\in\mathbb{R}^{d\times d}$.
Thus,
the Gibbs generalisation gap is given by
\begin{align}
    \Ep[\mathcal{G}_{\mathrm{G},n}]-\Ep[\mathcal{E}_{\mathrm{G},n}]=\frac{2}{n}\frac{n\beta\tau}{n\beta\tau+1}\mathrm{tr}(A),
    \label{eq: case study gap}
\end{align}
while the posterior covariance is given by
\begin{align}
    \frac{1}{n}\sum_{i=1}^{n}\Covpos[\nu(X_{i},\theta),s(X_{i},\theta)]
    =-\frac{2}{n}\frac{n\beta\tau}{n\beta\tau+1}
    \frac{\sum_{i=1}^{n}\widetilde{X}_{i}^{\top}
    A\widetilde{X}_{i}}{n}
    -
    \frac{\beta}{(n\beta+1/\tau)^{2}}\mathrm{tr}(A),
    \label{eq: case study posterior covariance}
\end{align}
where $\widetilde{X}_{i}:=X_{i}-\hat{\theta}$ and the detailed calculi are given in Appendix \ref{appendix: detailed calculi}.
Thus, we obtain
\begin{align}
    &\Ep[\mathcal{G}_{\mathrm{G},n}]-\Ep[\mathcal{E}_{\mathrm{G},n}]
    \nonumber\\
    &=-\Ep\left[\frac{1}{n}\sum_{i=1}^{n}\Covpos[\nu(X_{i},\theta),s(X_{i},\theta)]\right]
    +\frac{2}{n}\frac{n\beta\tau}{n\beta\tau+1}
    \frac{(\theta^{*})^{\top}A\theta^{*}}{(n\beta\tau+1)^{2}}
    +\mathrm{rem},
    \label{eq: case study difference}
\end{align}
where we let
\begin{align*}
    \mathrm{rem}=
    \frac{\beta\mathrm{tr}(A)}{(n\beta+1/\tau)^{2}}
    +\left\{\left(\frac{n\beta\tau}{n\beta\tau+1}\right)^{3}-2\left(\frac{n\beta\tau}{n\beta\tau+1}\right)^{2}\right\}\frac{2\mathrm{tr}(A)}{n^{2}}.
\end{align*}

\begin{figure}[h]
    \centering
    \includegraphics[width=110mm]{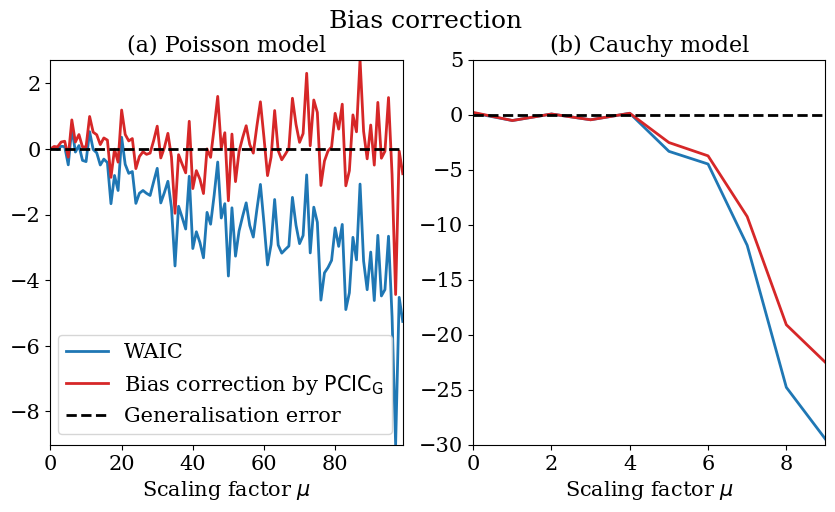}
    \caption{Bias correction of WAIC for the Poisson model (a) and the multivariate Cauchy model (b). The vertical axes denote the values relative to the generalisation errors, while the holizontal axes denote the scaling factor to the true values.}
    \label{fig:bias_correction}
\end{figure}

From (\ref{eq: case study difference}), we conclude that
in simple location-shift models,
under the assumption
\begin{align}
(\theta^{*})^{\top}A\theta^{*}/(n\beta\tau+1)^{2}=o(1),
\label{eq: assumption bias}
\end{align}
even the naive
$\mathrm{PCIC}_{\mathrm{G}}$ estimates well the Gibbs generalisation error regardless of the dimension $d$ and the distribution of the error term.
For non-strong priors ($\tau=O(1)$),
we can make the assumption in (\ref{eq: assumption bias}).
For strong priors ($1/\tau \sim n$),
we cannot expect the validity of the assumption, and the bias term \[\frac{2}{n}\frac{n\beta\tau}{n\beta\tau+1}
    \frac{(\theta^{*})^{\top}A\theta^{*}}{(n\beta\tau+1)^{2}}\]
remains.
Our generalised Bayesian framework provides a simple modification to remove this bias. Add the log-prior density (up to constant) divided by $n$ to $s(x,\theta)$;
\begin{align*}
s'(x,\theta)
&=s(x,\theta) + (1/n)\log \pi(\theta)\\
&=s(x,\theta)-(1/n)\|\theta\|^{2}/(2\tau) -(d/2)\log (2\pi \tau).
\end{align*}
Then, we have
\begin{align}
    \Ep[\mathcal{G}_{\mathrm{G},n}]
    -\Ep[\mathcal{E}_{\mathrm{G},n}]=-\Ep\left[\frac{1}{n}\sum_{i=1}^{n}\Covpos[\nu(X_{i},\theta),s'(X_{i},\theta)]\right]+\mathrm{rem}_{2}
    \label{eq: case study modification}
\end{align}
with $\mathrm{rem}_{2}$ an $O(n^{-2})$-term that is independent of $\theta^{*}$, which implies that 
the bias term in the naive $\mathrm{PCIC}_{\mathrm{G}}$ can be removed 
by changing the naive average posterior covariance to
\[
\frac{1}{n}\sum_{i=1}^{n}\Covpos\left[\nu(X_{i},\theta)\,,\,
s(X_{i},\theta)+\frac{1}{n}\log \pi(\theta)\right].
\]

\textbf{A scheme for general models}:
The above scheme is applicable in eliminating the bias of $\mathrm{WAIC}_{2}$ in general models by changing the average posterior variance to
\begin{align*}
    \frac{1}{n}\sum_{i=1}^{n}\Covpos\left[\log p(X_{i}\mid\theta)\,,\,\log p(X_{i}\mid\theta)+\frac{1}{n}\log \pi(\theta) \mid X^{n}\right].
\end{align*}
We check the effectiveness of this scheme by using the Poisson model and the Cauchy model.
Figure \ref{fig:bias_correction} demonstrates the results for the Poisson model with the conjugate Gamma prior
and the multivariate Cauchy model with the Cauchy prior; that is,
the observation and the prior in the Poisson model are described as
\[
X \sim \mathrm{Po}(\lambda_{0}) \quad \text{and}\quad  
\lambda \sim \mathrm{Gamma}(1,1),
\]
while those in the Cauchy model are described as
\[
X\sim \mathrm{Multivariate}\,\mathcal{C}(\mu_{0} \,,\,I_{5\times 5})
\quad \text{and}\quad
\mu \sim \mathrm{Multivariate}\,\mathcal{C}(0\,,\,I_{5\times 5}).
\]
After setting $\lambda_{0}=\mu$ and $\mu_{0}=\mu\mathbbm{1}$,
we vary the scaling factor $\mu$.
In the Poisson conjugate model,
the bias correction successfully works as depicted in Figure \ref{fig:bias_correction} (a).
In the Cauchy model,
the bias correction does work but there appears an additional non-negligible bias. Note that such bias also appears with large data scaling in the application to Bayesian hierarchical models; see Figure \ref{fig:BHM_datascale}.

\section{Discussions}\label{sec:discussions}

In this section, we discuss the following three topics:
\begin{itemize}
    \item the applicability in high dimensions;
    \item the application to defining a case-influence measure; and
    \item possibility of different definitions of generalisation errors.
\end{itemize}

\subsection{Applicability in high dimensions}
\label{subsec:highdimension}

We first check the applicability of our methodologies as well as IS-CV and PSIS-CV in high dimensional models.
If we view the location-shift model 
in Subsection \ref{subsec: bias reduction}
from a different perspective, 
the arguments in that subsection suggest that PCIC can be applied even in high-dimensional settings. This does not contradict the results of \cite{Okuno_Yano},
where WAIC is shown to work even in overparameterised linear regression models.
However, since these fully rely on the model's structures, further experiments are required.

Here, 
we employ the Poisson sequence model, a canonical model for count-data analysis. 
Incorporating the high-dimensional structure with Poisson sequence model becomes important in recent count-data analysis
(\citealp{komaki2006b,datta2016,yanokanekokomaki2021,hamuraetal2022,paul2024}).
The working model here is 
\[
Y(i)=\begin{pmatrix}
Y_{1}(i)\\
Y_{2}(i)\\
\vdots\\
Y_{d}(i)
\end{pmatrix} \quad \mid \quad \lambda =  \begin{pmatrix}
\lambda_{1}\\
\lambda_{2}\\
\vdots\\
\lambda_{d}
\end{pmatrix}
\quad\sim\quad
\otimes_{j=1}^{d} \mathrm{Poisson}(\lambda_{j})
\quad (i=1,\ldots,n),
\]
where $i$ is an index for the observation, $j$ is an index for the coordinate,
and $\otimes$ denotes the product of measures.
In this model,
each observation $Y(i)$
given $\lambda$
follows from the $d$-dimensional independent Poisson distribution.
For example, in spatio-temporal count-data analysis,
$i$ may be the index for the year and $j$ may be the index for the observation site such as a district; see \cite{datta2016,yanokanekokomaki2021,hamuraetal2022} for the details.

\begin{figure}[h]
    \centering
    \includegraphics[width=120mm]{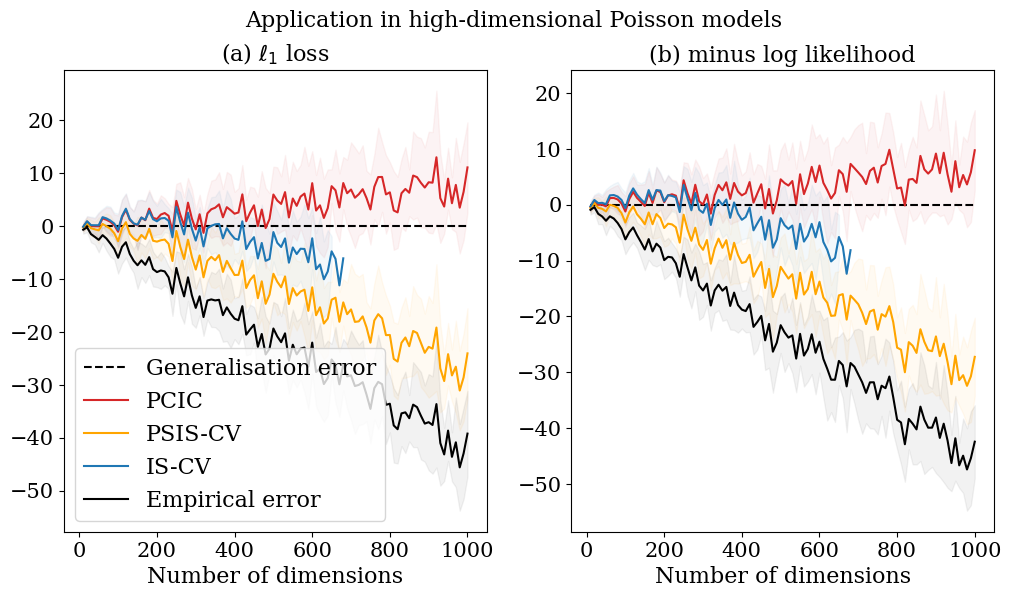}
    \caption{
    Application of $\mathrm{PCIC}_{\mathrm{G}}$, PSIS-CV, and IS-CV to high-dimensional Poisson models. Subfigure (a) displays the results for $\ell_{2}$ loss; Subfigure (b) displays the results for $\ell_{2}$ loss, where $\mathrm{PCIC}_{\mathrm{G}}$ is equal to $\mathrm{WAIC}_{2}$ in this case.
    }
    \label{fig: high dimension}
\end{figure}

\begin{figure}[h]
    \centering
    \includegraphics[width=100mm]{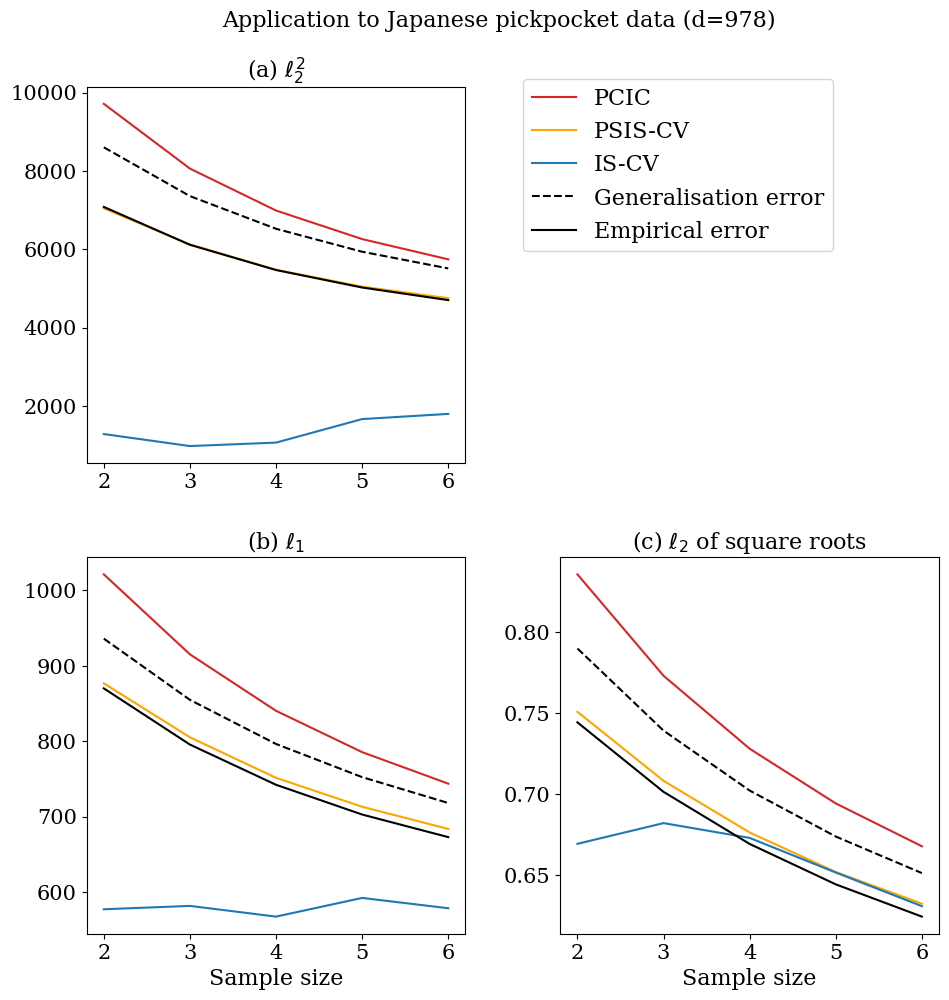}
    \caption{
    Application of $\mathrm{PCIC}_{\mathrm{G}}$, PSIS-CV, and IS-CV to the Japanese pickpocket data $(d=978)$. 
    The vertical axes exhibit values of the generalisation error estimates, while the horizontal axes depict the sample size.
    Subfigure (a) displays the results for $\ell_{2}^{2}$ loss; Subfigure (b) displays the results for $\ell_{1}$ loss;
    Subfigure (c) displays the results for $\ell_{1}$ loss of square roots ($\nu_{\mathrm{sq}}$).
    }
    \label{fig: real high dimension}
\end{figure}

We start with checking the behaviours of generalisation error estimates along with the change of the number $d$ of dimensions.
We fix the sample size $n$ to 10.
In this numerical experiment,
we work with the Stein-type shrinkage prior proposed by \cite{komaki2006b}:
\begin{align*}
\pi(\lambda)&= \frac{\lambda_{1}^{\beta_{1}-1}\cdots\lambda_{d}^{\beta_{d}-1}}{(\lambda_{1}+\cdots+\lambda_{d})^{\alpha}},
\end{align*}
where $\alpha>0$ and $\beta=(\beta_{1},\ldots,\beta_{d})$.
The reason of this prior choice is that 
it is non-informative, 
efficiently combines information across coordinates, 
and has an efficient exact sampling algorithm.
We set $\beta=(3,\ldots,3)$ and $\alpha=\sum_{j=1}^{d}\beta_{j}-1$.
For the true values of $\lambda_{j}$s, we set
\begin{align*}
\lambda_{j}&= 
\begin{cases}
0.001 &\text{if}\,\,j \,\,\text{is odd},\\
 2 &\text{if otherwise}
\end{cases}
\quad
(j=1,\ldots,d).
\end{align*}

Figure \ref{fig: high dimension} depicts the success of $\mathrm{PCIC}_{\mathrm{G}}$ even in high-dimensional models,
though we have not ensured it theoretically.
\cite{Okuno_Yano} shows the success of $\mathrm{WAIC}_{2}$ even in overparameterized linear regression models;
\cite{Giordano_Broderick_arXiv} discusses the Bayesian infinitesimal jackknife approximation of the variance in high dimensions.
These together with our result might deliver a new research direction focusing on understanding the behaviour of the posterior covariance in high-dimensions.
Understanding this might require further theoretical investigations on the approximation of cross validation in high dimensions (e.g., \citealp{RadandMaleki_2020})
and 
Bayesian central limit theorem in high-dimensions (e.g., \citealp{panov2015,yano2020,Kasprzaki2023}).

We next check the behaviours in the application to real data. We employ Japanese pickpocket data from \cite{pickpocket}. This data reports the total numbers of pickpockets in each year in Tokyo Prefecture, and are classified by town and also by the type of crimes. We use pickpocket data from 2012 to 2018 at 978 towns in 8 wards;
So, we work with the Poisson sequence model ($d=978$, $n= 7$; $n$ is the number of years we use in the analysis).
For $N=2,3,4,5,6$, 
we consider all combinations of selecting $N$ out of $n=7$ years, and use each combination as training data while treating the remaining items as test data. For each training and test pair, we calculate the generalisation error as well as its estimates ($\mathrm{PCIC}_{\mathrm{G}}$, PSIS-CV, IS-CV, and the empirical error), and take their averages.

Figure \ref{fig: real high dimension} displays the results for three loss functions: $\ell_{2}^{2}$ loss, $\ell_{1}$ loss, and
\[
\nu_{\mathrm{sr}}(Y(i),\lambda)=
\sum_{j=1}^{d}\left|\sqrt{Y_{j}(i)}-\sqrt{\lambda_{j}}\right|.
\]
For all loss functions, $\mathrm{PCIC}_{\mathrm{G}}$ estimates the generalisation error well.

\subsection{Application to defining a case-influence measure}

\begin{figure}[h]
    \centering
    \includegraphics[width=100mm]{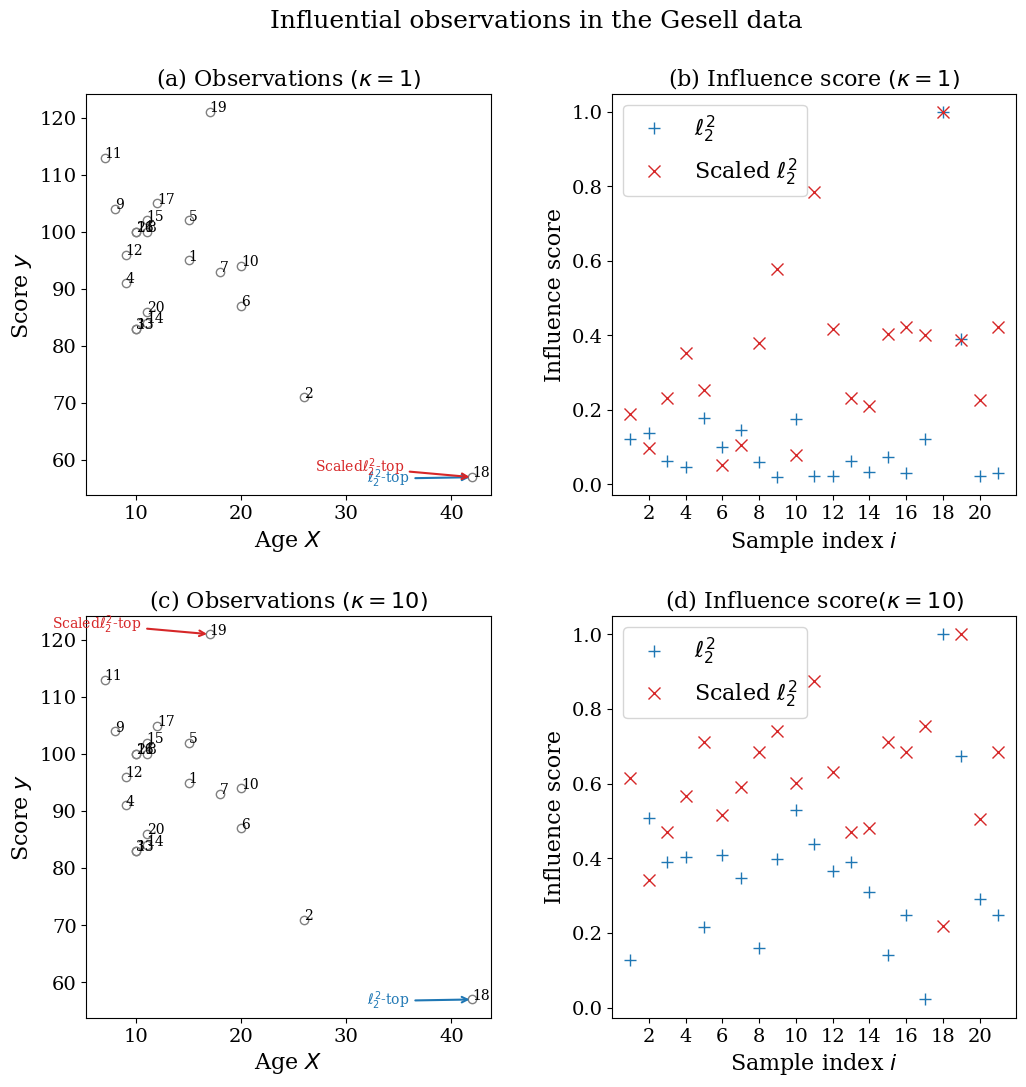}
    \caption{
    Influential observations in the Gesell data
    via $M_{\nu}$ for the generalisation. The values are divided by their maximum. 
    Subfigures (a) and (b) exhibit results for the posterior with the hyperparameter $\kappa$ in equation (\ref{eq:Bayesian linear regression}) set to be $1$.
    Subfigures (c) and (d) exhibit those for $\kappa$ in equation (\ref{eq:Bayesian linear regression}) set to $10$.
    }
    \label{fig:Gesell_sensitivity}
\end{figure}

We next consider an application of the proposed method to defining yet another case-influence measure.
From the perspective of estimating the generalisation error, 
the posterior covariance
\begin{align*}
    M_{\nu,i} := \Covpos[\nu(X_{i},\theta),s(X_{i},\theta)].
\end{align*}
expresses the contribution of each observation in filling the generalisation gap.
So, we can utilise it to define
yet another Bayesian case-influence measure in prediction.
When $s(\cdot,\cdot)$ and $\nu(\cdot,\cdot)$,
this coincides with the Bayesian local case sensitivity defined in \cite{MillarandStewart(2007),Thomas_MacEachern_Peruggia_2018}.
Note that this is not always positive except for $\nu(\cdot,\cdot)=s(\cdot,\cdot)$; so we standardize it and take the absolute value.

Figure \ref{fig:Gesell_sensitivity} showcases how the measure $M_{\nu}$ works in the Gesell data discussed in Subsection \ref{subsec: influential observations}, where we use the unscaled $\ell_{2}^{2}$ and the scaled $\ell_{2}^{2}$ losses,
and the score function $s(x,\theta)$ is log likelihood.
For relatively light-tailed prior ($\kappa=1$), the top-1 influential observations are the same for both losses (Figure \ref{fig:Gesell_sensitivity} a) .
For relatively heavy-tailed prior ($\kappa=10$), the top-1 influential observations are different (Figure \ref{fig:Gesell_sensitivity} c).

Here we compare our results with those in \cite{MillarandStewart(2007)}, where two types of sensitivities (parameter-scale and observation-scale) are defined.
First, our results with the unscaled loss are similar to the results with the observation-scale sensitivity in \cite{MillarandStewart(2007)},
and the results with the scaled loss
are similar to the results with the parameter-scale sensitivity in \cite{MillarandStewart(2007)};
Second, our results with the relatively light-tailed prior are similar to \cite{MillarandStewart(2007)}'s results in the known-variance case,
and our results with the relatively heavy-tailed prior are similar to
\cite{MillarandStewart(2007)}'s results 
the unknown-variance case with Bernard's reference prior.
Our result as well as \cite{MillarandStewart(2007)} and \cite{Thomas_MacEachern_Peruggia_2018} suggest that the case influences also vary depending on the choice of loss functions.
Also, our result as well as \cite{MillarandStewart(2007)} suggest that the case influences vary depending on the prior specification.

\subsection{Possibility of different definitions of generalisation errors}

Although we only consider the Gibbs and plug-in generalization errors, following the literature (e.g.,~\citealp{Ando_2007, Underhill_Smith_2016}), one could alternatively define the generalization error as
\[
\mathcal{G}_{\mathrm{S},n} 
= \Ep_{X_{n+1}} \Bigl[ S\bigl(X_{n+1},\, \Eppos[p(\cdot \mid \theta)]\bigr)\Bigr],
\]
where \(S(\cdot,\cdot)\) is a scoring rule \citep{Gneiting_2007, Dawid_etal_2016}, and \(\{p(x \mid \theta)\}\) is a parametric model. This may be viewed as a natural extension of the generalization error used in WAIC \citep{Watanabe_2010}.

In the binary prediction setting discussed in Subsection~\ref{subsec:privacy}, the plug-in generalization errors coincide with this definition. However, our current methodologies do not accommodate this scoring-rule-based generalisation error in more general prediction scenarios. In this paper, we have not pursued this direction because it remains unclear whether the use of the Bayesian predictive density \(\Eppos[p(\cdot \mid \theta)]\) is theoretically justified. While for the log score, the Bayesian predictive density is the Bayes solution \citep{Aitchison_1975}, it is not necessarily the Bayes solution for other loss functions (see, e.g., \citealp{Corcuera_Giummole_1999}).

Recently, a quasi-Bayesian framework associated with scoring rules has been developed \citep{Giummole_etal_2019,Matsubara_JASA}. Thus, extending our understanding of scoring-rule-based generalization errors in conjunction with such quasi-Bayesian theory would be an important future research direction.

\section{Conclusion} \label{sec:conclusion}

We have proposed a novel, computationally low-cost method of estimating the Gibbs generalisation errors and plugin generalisation errors for arbitrary loss functions.
We have demonstrated the usefulness of the proposed method 
in privacy-preserving learning,
Bayesian hierarchical modeling,
Bayesian regression modeling in the presence of the influential observations,
and
discussed the applicability in high dimensional statistical models.

Our numerical experiments suggest that IS-CV, PSIS-CV, and PCIC are useful tools for assessing generalisation errors.
However, when applied to high-dimensional models (Subsection \ref{subsec:highdimension}) or when the magnitude of the data is large (Subsection \ref{subsec:BHM}), PCIC remains effective even under such settings.
Detailed theoretical analyses remain an important direction for future research.

An important practical implication of this study is that the posterior covariance provides an easy-to-implement generalisation error estimate for arbitrary loss functions, and can avoid the cumbersome refitting in LOOCV as well as the importance sampling technique that is sensitive to the presence of influential observations.
Also, theoretical connections between WAIC, the Bayesian sensitivity analysis, and the infinitesimal jackknife approximation of Bayesian LOOCV are clarified. \revision{by our proof for the asymptotic unbiasedness}.

\section*{Acknowledgement}

The authors would like to thank the handling editor, the associate editor, and two anonymous referees for their comments that have improved the quality of the paper. 
Also, the authors would like to thank 
Akifumi Okuno, Hironori Fujisawa,
Yoshiyuki Ninomiya, and Yusaku Ohkubo for providing them with fruitful discussions.

\section*{Funding}
This work was supported by 
Japan Society for the Promotion of Science (JSPS) [Grant Nos. 19K20222, 21H05205, 21K12067],
the Japan Science and Technology Agency's Core Research for Evolutional Science and Technology (JST CREST) [Grant No. JPMJCR1763], and the MEXT Project for Seismology toward Research Innovation with Data of Earthquake (STAR-E) [Grant No. JPJ010217].
\reviseend

\section*{Code availability statements}

The python code is available at \url{https://github.com/kyanostat/PCIC4GL}.

\bibliographystyle{rss}
\bibliography{mcmcoptimism}

\appendix

\setcounter{equation}{0} 
\setcounter{figure}{0} 
\renewcommand{\theequation}{\thesection.\arabic{equation}}
\renewcommand{\thefigure}{\thesection.\arabic{figure}}

\section{Conditions for the theorem}\label{appendix: assumptions}

The conditions in this paper are as follows:
In the conditions,
$\nu$ is redefined by subtracting 
$\min_{\theta'\in\Theta}\Ep[\nu(X_{n+1},\theta')]$ from the original $\nu$.
\begin{enumerate}
    \item[(C1)]
    The difference
    $\{\Ep_{X^{n-1}}[\mathcal{G}_{\mathrm{G},n-1}]-\Ep[\mathcal{G}_{\mathrm{G},n}]\}$ is of the order $o(\Ep[\mathcal{G}_{\mathrm{G},n}])$;
    \item[(C2)] The following relation holds:
    \begin{align*}
        &\Ep\left[
        \sup_{w\in \mathcal{W}^{(-1)}}
        \left|
        \Eppos^{w}\left[
        \left\{\nu(X_{1},\theta)-\Eppos^{w}[\nu(X_{1},\theta)]\right\}
    \left\{s(X_{1},\theta)-\Eppos^{w}[s(X_{1},\theta)]\right\}^{2}
        \right]
        \right|
        \right]\\
        &=o(\Ep[\mathcal{G}_{\mathrm{G},n}]);
    \end{align*}
    \item[(C3)] There exists an integrable function $M(\theta)$ such that
    for all $w\in \mathcal{W}^{(-1)}$, 
    \begin{align*}
    &\left| \pi(\theta)
    \exp\left\{
    w_{1}s(X_{1},\theta)+
    \sum_{k\neq 1}s(X_{k},\theta) \right\}
    \nu^{l}(X_{1},\theta)
    s^{j}(X_{1},\theta)\right|\\
    &\le M(\theta),\quad l=0,1,\,j=0,1.
    \end{align*}
\end{enumerate}

\begin{remark}[Discussion on the conditions]
First, consider Condition C1.
In regular statistical models and smooth evaluation functions,
the result of \cite{Underhill_Smith_2016} implies \revision{that the order of $\Ep[\mathcal{G}_{G,n}]$ is $1/n$.}{}
that the expected generalisation error subtracted by its minimum  with respect to the parameter
    \[
    \tilde{\nu}(X_{i},\theta)
    =\nu(X_{i},\theta)-\min_{\theta'\in\Theta}\Ep[\nu(X_{n+1},\theta')].
    \]
is of order $1/n$.
We briefly discuss on it.
Let $\theta^{*}$ be the minimizer of $\Ep[\nu(X_{n+1},\theta')]$. 
The Taylor expansion yields
    \begin{align*}
    \Eppos[\nu(X_{n+1},\theta)\mid X^{n}]
    &=
    \Eppos[\nu(X_{n+1},\theta^{*})\mid X^{n}]+
    \nabla_{\theta^{*}}\nu(X_{n+1},\theta)
    \Eppos[\theta-\theta^{*}\mid X^{n}]\\
    &\, +\frac{1}{2}
    \Eppos[(\theta-\theta^{*})^{\top}\nabla_{\theta^{*}}^{2}
    \nu(X_{n+1},\theta)(\theta-\theta^{*})\mid X^{n}] + o_{P}(\|\theta-\theta^{*}\|^{2}).
    \end{align*}
    Thus, we get
    \begin{align*}
    \Eppos[\tilde{\nu}(X_{n+1},\theta)\mid X^{n}]
    &=
    \underbrace{\nabla_{\theta^{*}}\nu(X_{n+1},\theta)
    \Eppos[\theta-\theta^{*}\mid X^{n}]}_{=:T_{1}} \\
    &\, +\frac{1}{2}
    \underbrace{\Eppos[(\theta-\theta^{*})^{\top}\nabla_{\theta^{*}}^{2}
    \nu(X_{n+1},\theta)(\theta-\theta^{*})\mid X^{n}]}_{=:T_{2}} + o_{P}(\|\theta-\theta^{*}\|^{2}).    
    \end{align*}
    The order of the term $T_{1}$ is controlled by the difference between the posterior mean and $\theta^{*}$; so it is of order $1/n$.
    The order of the term $T_{2}$ is determined by Lemma 3 of Underhill and Smith (2016);
    so it is of order $1/n$,
    which concludes the order of the expected (subtracted) generalisation error.
    For singular statistical models and log-likelihoods,
the result of \cite{Watanabe_2010}
implies that the order of $\Ep[\mathcal{G}_{G,n}]$ (subtracted by its minimum) is $1/n$. Thus, Condition C1 usually holds.

Condition C2 is a condition for the residual. For regular statistical models and smooth evaluation functions,
the Taylor expansion implies that
the order of the left hand side is $n^{-3/2}$ and less than the order of $\Ep[\mathcal{G}_{G,n}]$. 
For singular statistical models and log-likelihoods,
we refer to \cite{Watanabe_book_mtbs}.
Finally, Condition C3 is a mild condition for assuming the existence of expectation.
\end{remark}

\section{Proof for Lemma \ref{lem:sensitivity}}\label{appendix: proof of lemma}
This appendix provides the proof of Lemma \ref{lem:sensitivity}.

The Lebesgue dominated convergence theorem ensures the exchange of differentiation and integration in $(\partial^{k}/\partial w_{i}^{k})\Eppos^{w}[\nu(X_{i},\theta)]$
under Condition C3. Then, considering 
\[
F(w):=\int \exp\left\{\sum_{i=1}^{n}w_{i}s(X_{i},\theta)\right\}\pi(\theta)d\theta
\]
gives
\begin{align*}
    \frac{\partial}{\partial w_{i}}\Eppos^{w}[\nu(X_{i},\theta)]
    &=
    \frac{\int s(X_{i},\theta)\nu(X_{i},\theta)\mathrm{e}^{\sum_{i=1}^{n}w_{i}s(X_{i},\theta)}\pi(\theta)d\theta F(w)
    }
    {F^{2}(w)}\\
    &\quad -
    \frac{\int s(X_{i},\theta)\mathrm{e}^{\sum_{i=1}^{n}w_{i}s(X_{i},\theta)}\pi(\theta)d\theta
    \int \nu(X_{i},\theta)\mathrm{e}^{\sum_{i=1}^{n}w_{i}s(X_{i},\theta)}\pi(\theta)d\theta}{F^{2}(w)},
\end{align*}
which implies
\begin{align*}
        \frac{\partial}{\partial w_{i}}
    \Eppos^{w}[\nu(X_{i},\theta)]
    =\Eppos^{w}\left[\{\nu(X_{i},\theta)-\Eppos^{w}[\nu(X_{i},\theta)]\}\{s(X_{i},\theta)-\Eppos^{w}[s(X_{i},\theta)]\}\right].
\end{align*}
Further, 
we have
\begin{align*}
    &\frac{\partial}{\partial w_{i}}\frac{\int s(X_{i},\theta)\nu(X_{i},\theta)\mathrm{e}^{\sum_{i=1}^{n}w_{i}s(X_{i},\theta)}\pi(\theta)d\theta F(w)
    }
    {F^{2}(w)}\\
    &=
    \frac{\int s^{2}(X_{i},\theta)\nu(X_{i},\theta)\mathrm{e}^{\sum_{i=1}^{n}w_{i}s(X_{i},\theta)}\pi(\theta)d\theta}{F(w)}\\
    &\quad -\frac{\int s(X_{i},\theta)\nu(X_{i},\theta)\mathrm{e}^{\sum_{i=1}^{n}w_{i}s(X_{i},\theta)}\pi(\theta)d\theta
    \int s(X_{i},\theta)\mathrm{e}^{\sum_{i=1}^{n}w_{i}s(X_{i},\theta)}\pi(\theta)d\theta
    }{F^{2}(w)}\\
    &=
    \Eppos^{w}[s^{2}(X_{i},\theta)\nu(X_{i},\theta)]
    -
    \Eppos^{w}[s(X_{i},\theta)\nu(X_{i},\theta)]
    \Eppos^{w}[s(X_{i},\theta)]
\end{align*}
and 
\begin{align*}
    &\frac{\partial}{\partial w_{i}}\frac{\int s(X_{i},\theta)\mathrm{e}^{\sum_{i=1}^{n}w_{i}s(X_{i},\theta)}\pi(\theta)d\theta
    \int \nu(X_{i},\theta)\mathrm{e}^{\sum_{i=1}^{n}w_{i}s(X_{i},\theta)}\pi(\theta)d\theta}{F^{2}(w)}\\
    &=\frac{\int s^{2}(X_{i},\theta)\mathrm{e}^{\sum_{i=1}^{n}w_{i}s(X_{i},\theta)}\pi(\theta)d\theta \int \nu(X_{i},\theta)\mathrm{e}^{\sum_{i=1}^{n}w_{i}s(X_{i},\theta)}\pi(\theta)d\theta F^{2}(w)}{F^{4}(w)}\\
    &\quad+\frac{\int s(X_{i},\theta)\mathrm{e}^{\sum_{i=1}^{n}w_{i}s(X_{i},\theta)}\pi(\theta)d\theta \int \nu(X_{i},\theta)s(X_{i},\theta)\mathrm{e}^{\sum_{i=1}^{n}w_{i}s(X_{i},\theta)}\pi(\theta)d\theta F^{2}(w)}{F^{4}(w)}\\
    &\quad-2\frac{\{\int s(X_{i},\theta)\mathrm{e}^{\sum_{i=1}^{n}w_{i}s(X_{i},\theta)}\pi(\theta)d\theta\}^{2}
    \int \nu(X_{i},\theta)\mathrm{e}^{\sum_{i=1}^{n}w_{i}s(X_{i},\theta)}\pi(\theta)d\theta  F(w)}{F^{4}(w)}\\
    &=\Eppos^{w}[s^{2}(X_{i},\theta)]\Eppos^{w}[\nu(X_{i},\theta)]
    +\Eppos^{w}[s(X_{i},\theta)\nu(X_{i},\theta)]\Eppos^{w}[s(X_{i},\theta)]\\
    &\quad-2(\Eppos^{w}[s(X_{i},\theta)])^{2}\Eppos^{w}[\nu(X_{i},\theta)],
\end{align*}
which implies 
\begin{align*}
        \frac{\partial^{2}}{\partial w_{i}^{2}}
    \Eppos^{w}[\nu(X_{i},\theta)]
    =\Eppos^{w}\left[\{\nu(X_{i},\theta)-\Eppos^{w}[\nu(X_{i},\theta)]\}\{s(X_{i},\theta)-\Eppos^{w}[s(X_{i},\theta)]\}^{2}\right]
\end{align*}
and completes the proof.
\qed

\section{Detailed calculations used in Section \ref{subsec: bias reduction}}
\label{appendix: detailed calculi}
This appendix provides the detailed calculation used in Section \ref{subsec: bias reduction}.
The calculation employs the following lemma.
\begin{lemma}[Kumar, 1973]
\label{lem: product of quadratic forms}
Let $w$ be a random vector from $N(0,I_{d})$.
For $d\times d$ symmetric matrices $B$ and $C$, we have
\begin{align*}
    \Ep[ (w^{\top}Bw)(w^{\top}Cw)]
    =2\mathrm{tr}(BC)+(\mathrm{tr}B)\,(\mathrm{tr}C).
\end{align*}
\end{lemma}

\textit{Step 1. Establishing (\ref{eq: case study posterior covariance})}:
Let us begin by expanding $\Covpos[s(X_{i},\theta),\nu(X_{i},\theta)]$.
For $i=1,\ldots,n$, we have
\begin{align*}
    &\Eppos[(X_{i}-\theta)^{\top}A(X_{i}-\theta)(X_{i}-\theta)^{\top}(X_{i}-\theta)]\\
    &=\Eppos\Big[
    \Big\{(X_{i}-\hat{\theta})^{\top}A(X_{i}-\hat{\theta})
    +2(\hat{\theta}-\theta)^{\top}A(X_{i}-\hat{\theta})
    +(\hat{\theta}-\theta)^{\top}A(\hat{\theta}-\theta)
    \Big\}\\
    &\quad\quad\quad\quad\quad\quad
    \Big\{
    (X_{i}-\hat{\theta})^{\top}(X_{i}-\hat{\theta})
    +2(\hat{\theta}-\theta)^{\top}(X_{i}-\hat{\theta})
    +(\hat{\theta}-\theta)^{\top}(\hat{\theta}-\theta)
    \Big\}
    \Big]\\
    &=\widetilde{X}_{i}^{\top}A\widetilde{X}_{i}
    \widetilde{X}_{i}^{\top}\widetilde{X}_{i}
    +
    \Eppos\big[2\widetilde{X}_{i}^{\top}A\widetilde{X}_{i}
    \widetilde{\theta}^{\top}\widetilde{X}_{i}\big]
    +
    \Eppos\big[\widetilde{X}_{i}^{\top}A\widetilde{X}_{i}
    \widetilde{\theta}^{\top}\widetilde{\theta}
    \big]
    \\
    &\quad
    +\Eppos\big[2\widetilde{\theta}^{\top}A\widetilde{X}_{i}
    \widetilde{X}_{i}^{\top}\widetilde{X}_{i}
    \big]
    +\Eppos\big[4\widetilde{\theta}^{\top}A\widetilde{X}_{i}
    \widetilde{\theta}^{\top}\widetilde{X}_{i}
    \big]
    +\Eppos\big[2\widetilde{\theta}^{\top}A\widetilde{X}_{i}
    \widetilde{\theta}^{\top}\widetilde{\theta}
    \big]
    \\
    &\quad
    +\Eppos\big[\widetilde{\theta}^{\top}A\widetilde{\theta}
    \widetilde{X}_{i}^{\top}\widetilde{X}_{i}
    \big]
    +\Eppos\big[2\widetilde{\theta}^{\top}A\widetilde{\theta}
    \widetilde{\theta}^{\top}\widetilde{X}_{i}
    \big]
    +\Eppos\big[\widetilde{\theta}^{\top}A\widetilde{\theta}
    \widetilde{\theta}^{\top}\widetilde{\theta}
    \big],
\end{align*}
where $\widetilde{X}_{i}:=X_{i}-\hat{\theta}$ and 
$\widetilde{\theta}:=\hat{\theta}-\theta$.
Lemma \ref{lem: product of quadratic forms} gives 
\begin{align*}
    \Eppos[(\hat{\theta}-\theta)^{\top}(\hat{\theta}-\theta)
    (\hat{\theta}-\theta)^{\top}A(\hat{\theta}-\theta)]
    = \frac{2+d}{(n\beta+1/\tau)^{2}}\mathrm{tr}(A)
\end{align*}
and thus we get
\begin{align}
    &\Eppos[(X_{i}-\theta)^{\top}(X_{i}-\theta)(X_{i}-\theta)^{\top}A(X_{i}-\theta)]\nonumber\\
    &=
    \widetilde{X}_{i}^{\top}A\widetilde{X}_{i}
    \widetilde{X}_{i}^{\top}\widetilde{X}_{i}
    +\frac{4+d}{n\beta+1/\tau}\widetilde{X}_{i}^{\top}A\widetilde{X}_{i}
    +\frac{\mathrm{tr}(A)}{n\beta+1/\tau}
    \widetilde{X}_{i}^{\top}\widetilde{X}_{i}
    + \frac{2+d}{(n\beta+1/\tau)^{2}}\mathrm{tr}(A)
    \label{eq: snu 1}
\end{align}
Further, for $i=1,\ldots,n$, we have
\begin{align}
    &\Eppos[(X_{i}-\theta)^{\top}A(X_{i}-\theta)]
    \Eppos[(X_{i}-\theta)^{\top}(X_{i}-\theta)]
    \nonumber\\
    &=\left\{\widetilde{X}_{i}^{\top}\widetilde{X}_{i}+\frac{d}{n\beta+1/\tau}\right\}
    \left\{\widetilde{X}_{i}^{\top}A\widetilde{X}_{i}+\frac{\mathrm{tr}(A)}{n\beta+1/\tau}\right\}\nonumber\\
    &=\widetilde{X}_{i}^{\top}\widetilde{X}_{i}
    \widetilde{X}_{i}^{\top}A\widetilde{X}_{i}
    +\frac{\mathrm{tr}(A)}{n\beta+1/\tau}\widetilde{X}_{i}^{\top}\widetilde{X}_{i}
    +\frac{d}{n\beta+1/\tau}
    \widetilde{X}_{i}^{\top}A\widetilde{X}_{i}
    +\frac{d}{n\beta+1/\tau}
    \frac{\mathrm{tr}(A)}{n\beta+1/\tau}.
    \label{eq: snu 2}
\end{align}
Combining (\ref{eq: snu 1}) and (\ref{eq: snu 2}) yields
\begin{align*}
    \Covpos[\nu(X_{i},\theta),s(X_{i},\theta)]
    =-\frac{\beta}{2}\left\{\frac{4}{n\beta+1/\tau}\widetilde{Y}_{i}^{\top}A\widetilde{Y}_{i}+
    \frac{2\mathrm{tr}(A)}{(n\beta+1/\tau)^{2}}
    \right\},
\end{align*}
which implies (\ref{eq: case study posterior covariance}).

\textit{Step 2. Establishing (\ref{eq: case study difference})}:
Next, we take the expectation of $(1/n)\sum_{i=1}^{n}(X_{i}-\hat{\theta})^{\top}A(X_{i}-\hat{\theta})$. Let $a:=(n\beta\tau)/(n\beta\tau+1)$.
Then, we have
\begin{align*}
    &\Ep\left[\frac{1}{n}\sum_{i=1}^{n}(X_{i}-\hat{\theta})^{\top}A(X_{i}-\hat{\theta})\right]\\
    &=\Ep\left[\frac{1}{n}\sum_{i=1}^{n}(X_{i}-\theta^{*})^{\top}A(X_{i}-\theta^{*})\right]
    +(1-a)^{2}(\theta^{*})^{\top}A\theta^{*}
    \\
    &\qquad\qquad
    +(a^{2}-2a)\Ep\left[(\bar{X}-\theta^{*})^{\top}A(\bar{X}-\theta^{*})\right]\\
    &=\mathrm{tr}(A)+(1-a)^{2}(\theta^{*})^{\top}A\theta^{*}
    +(a^{2}-2a)\frac{\mathrm{tr}(A)}{n},
\end{align*}
which implies (\ref{eq: case study difference}).

\textit{Step 3. Establishing (\ref{eq: case study modification})}:
Finally, we consider $\Covpos[(X_{i}-\theta)^{\top}A(X_{i}-\theta),\theta^{\top}\theta]$. For $i=1,\ldots,n$, we have
\begin{align*}
    &\Eppos[(X_{i}-\theta)^{\top}A(X_{i}-\theta)\theta^{\top}\theta]\\
    &=\Eppos
    \Big[\Big\{
    (X_{i}-\hat{\theta})^{\top}A(X_{i}-\hat{\theta})
    +(\hat{\theta}-\theta)^{\top}A(\hat{\theta}-\theta)
    +2(X_{i}-\hat{\theta})^{\top}A(\hat{\theta}-\theta)
    \Big\}\\
    &\qquad\qquad\qquad\qquad\qquad\qquad\qquad\qquad
    \Big\{
    (\hat{\theta}-\theta)^{\top}(\hat{\theta}-\theta)
    -2(\hat{\theta}-\theta)^{\top}\hat{\theta}
    +\hat{\theta}^{\top}\hat{\theta}
    \Big\}
    \Big]\\
    &=
    \frac{1}{n\beta+1/\tau}(X_{i}-\hat{\theta})^{\top}A(X_{i}-\hat{\theta})
    +(X_{i}-\hat{\theta})^{\top}A(X_{i}-\hat{\theta})\hat{\theta}^{\top}\hat{\theta}
    +\frac{\mathrm{tr}(A)}{n\beta+1/\tau}\hat{\theta}^{\top}\hat{\theta}\\
    &\quad
    +\frac{2+d}{(n\beta+1/\tau)^{2}}\mathrm{tr}(A)
    +\frac{4}{n\beta+1/\tau}(X_{i}-\hat{\theta})^{\top}A\hat{\theta}
\end{align*}
and we have 
\begin{align*}
    &\Eppos[(X_{i}-\theta)^{\top}A(X_{i}-\theta)]\Eppos[\theta^{\top}\theta]\\
    &=
    \Big\{(X_{i}-\hat{\theta})^{\top}A(X_{i}-\hat{\theta})+\frac{\mathrm{tr}(A)}{n\beta+1/\tau}\Big\}
    \Big\{\hat{\theta}^{\top}\hat{\theta}+\frac{1}{n\beta+1/\tau}
    \Big\}\\
    &=
    \frac{1}{n\beta+1/\tau}(X_{i}-\hat{\theta})^{\top}A(X_{i}-\hat{\theta})
    +(X_{i}-\hat{\theta})^{\top}A(X_{i}-\hat{\theta})\hat{\theta}^{\top}\hat{\theta}
    +\frac{\mathrm{tr}(A)}{n\beta+1/\tau}\hat{\theta}^{\top}\hat{\theta}\\
    &\qquad\qquad+\frac{\mathrm{tr}(A)}{(n\beta+1/\tau)^{2}}.
\end{align*}
Combining these yields
\begin{align*}
    \Covpos[(X_{i}-\theta)^{\top}A(X_{i}-\theta),\theta^{\top}\theta]
    =-\frac{4}{n\beta+1/\tau}(X_{i}-\hat{\theta})^{\top}A\hat{\theta}
    +\frac{1+d}{(n\beta+1/\tau)^{2}}\mathrm{tr}(A).
\end{align*}
This,
together with the identity
\begin{align*}
    \Ep\left[\frac{1}{n}\sum_{i=1}^{n}(X_{i}-\hat{\theta})^{\top}A\hat{\theta}
    \right]
    =a(1-a)\Ep\left[\overline{X}^{\top}A\overline{X}\right]
    =a(1-a)\{(\theta^{*})^{\top}A\theta^{*}
    +\mathrm{tr}(A)/n\},
\end{align*}
yields (\ref{eq: case study modification}).

\subsection*{References for this appendix}
\begin{enumerate}\setlength{\leftskip}{+0.5cm}
    \item[{[1]}] Kumar, A. (1973) Expectation of products of quadratic forms. \textit{Sankha}, \textbf{35}, 359--362.
\end{enumerate}

\section{Additional figures}

This appendix provides additional figures.
Figure \ref{fig:stackloss} depicts the result for the stack loss data; see Subsection 3.2 for details.
Figure \ref{fig:Gesell} depicts the result for the Gesell data; see Subsection 3.2 for details.

\begin{figure}[h!]
    \centering
    \includegraphics[width=90mm]{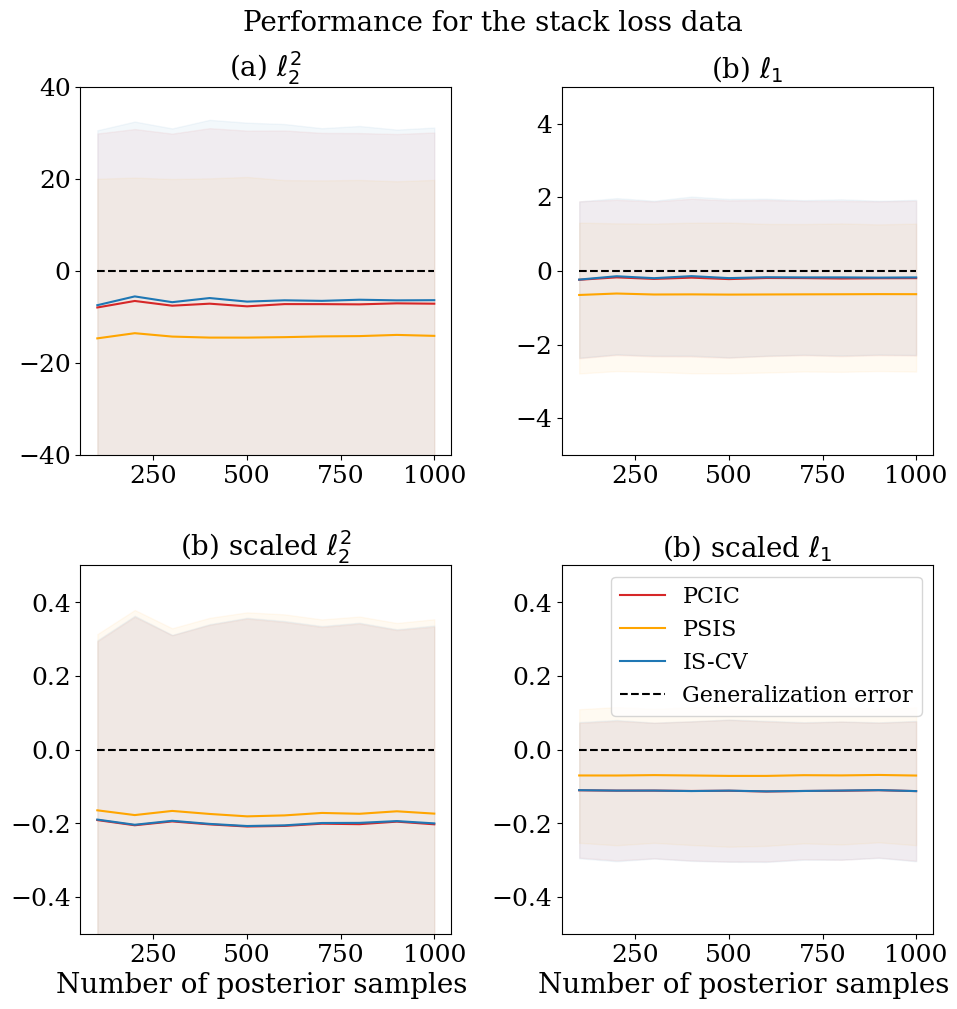}
    \caption{
    Performance of $\mathrm{PCIC}_{\mathrm{G}}$, IS-CV, and PSIS-CV for the stack loss dataset.
    The vertical axes present values relative to the generalisation errors,
    while the horizontal axes correspond to the number of posterior samples.
    Averages are computed across different experiments, with colored shades representing $\pm \sigma$.
    (a) results for $\ell_{2}^{2}$ loss;
    (b) those for $\ell_{1}$ loss;
    (c) those for scaled $\ell_{2}^{2}$ loss;
    (d) those for scaled $\ell_{1}$ loss.
    }
    \label{fig:stackloss}
\end{figure}

\begin{figure}[h!]
    \centering
    \includegraphics[width=90mm]{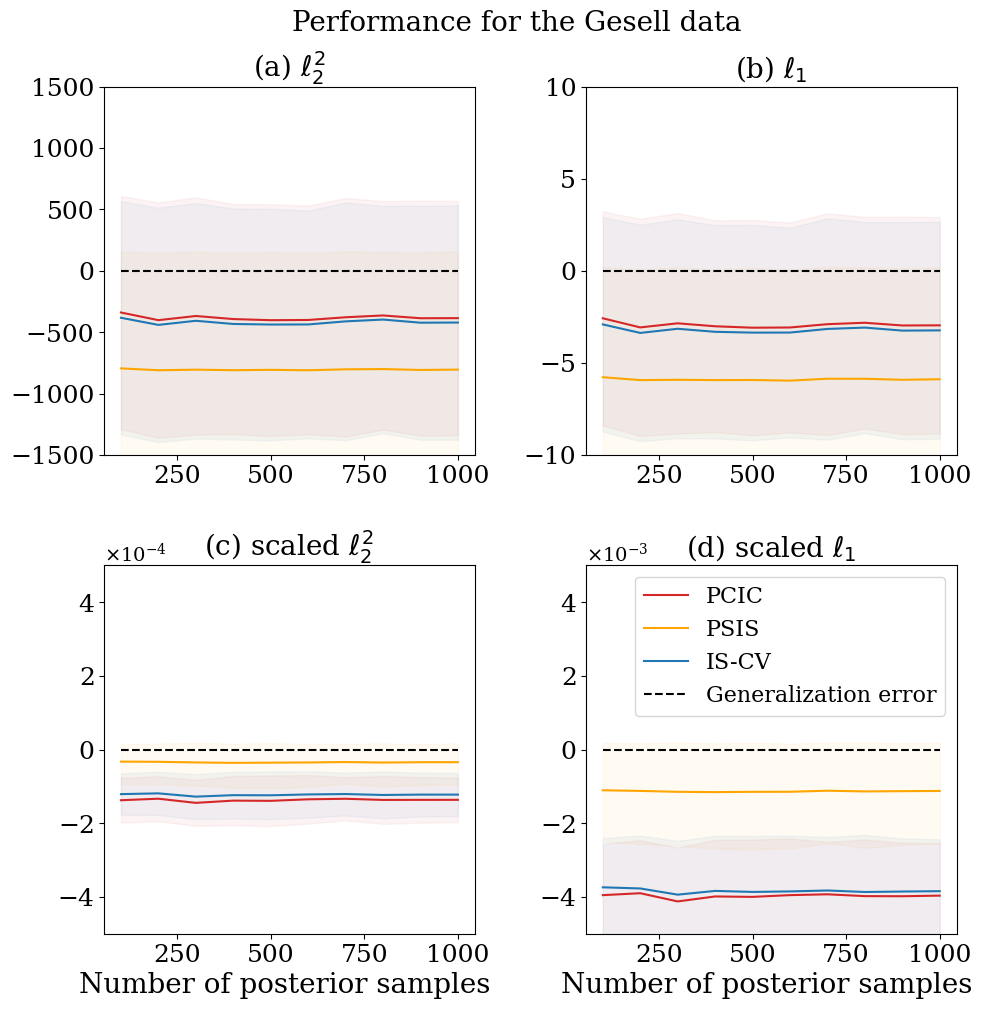}
    \caption{
    Performance of $\mathrm{PCIC}_{\mathrm{G}}$, IS-CV, and PSIS-CV for the Gesell dataset.
    The vertical axes present values relative to the generalisation errors,
    while the horizontal axes correspond to the number of posterior samples.
    Averages are computed across different experiments, with colored shades representing $\pm \sigma$.
    (a) results for $\ell_{2}^{2}$ loss;
    (b) those for $\ell_{1}$ loss;
    (c) those for scaled $\ell_{2}^{2}$ loss;
    (d) those for scaled $\ell_{1}$ loss.
    }
    \label{fig:Gesell}
\end{figure}

\end{document}